\documentclass[journal,11pt,doublespace,onecolumn]{IEEEtran}
\usepackage{amsmath}
\usepackage{graphicx,psfrag,epsf}
\usepackage{enumerate}
\usepackage{url} 
\usepackage{url} 
\usepackage{amssymb,amsfonts, subfigure,mathtools}
\usepackage{bm}
\usepackage{color, colortbl}
\usepackage{algorithm}
\usepackage{algorithmic}
\usepackage{balance}

\newtheorem{remark}{Remark}
\newtheorem{theorem}{Theorem}

\newtheorem{lemma}{Lemma} 

\newtheorem{proposition}{Proposition}

\newtheorem{proof}{Proof}

\def\T{{ \mathrm{\scriptscriptstyle T} }}
\def\P{{ \mathrm{pr} }}
\def\v{{\varepsilon}}

\def\m{{M_{\textrm{max}}}}

\DeclareMathOperator*{\argmin}{arg\,min}

\newcommand{\rom}[1]{\uppercase\expandafter{\romannumeral #1\relax}}

\begin{document}
%
\title{Multiple Change Point Analysis: Fast \\ Implementation And Strong Consistency}
%
%
%

\author{Jie~Ding,~\IEEEmembership{Student Member,~IEEE,}
        Yu~Xiang,~\IEEEmembership{Member,~IEEE,}
        Lu~Shen,~\IEEEmembership{Member,~IEEE,}
        and~Vahid~Tarokh,~\IEEEmembership{Fellow,~IEEE}
\thanks{This work is supported by Defense Advanced Research Projects Agency (DARPA) grant numbers W911NF-14-1-0508 and  N66001-15-C-4028.}
\thanks{J.~Ding, Y.~Xiang, Lu.~Shen, and V. Tarokh are with the John A. Paulson School of Engineering and Applied Sciences, Harvard University, Cambridge, 
MA, 02138 USA e-mail: (jieding@fas.harvard.edu).}
}

\markboth{}
{Shell \MakeLowercase{\textit{et al.}}: IEEE Transactions on Signal Processing}


\maketitle

\begin{abstract}
One of the main challenges in identifying structural changes in stochastic processes is to carry out analysis for time series with dependency structure in a computationally tractable way. Another challenge is that the number of true change points is usually unknown, 
 requiring a suitable model selection criterion to arrive at informative conclusions.  
To address the first challenge, we model the data generating process as a segment-wise autoregression, which is composed of several segments (time epochs), each of which modeled by an autoregressive model.
We propose a multi-window method that is both effective and efficient for discovering the structural changes.   
The proposed approach was motivated by transforming a segment-wise autoregression into a multivariate time series that is asymptotically segment-wise independent and identically distributed.  
To address the second challenge, we derive theoretical guarantees for (almost surely) selecting the true number of change points of segment-wise independent multivariate time series.
Specifically, under mild assumptions, we show that a Bayesian Information Criterion (BIC)-like criterion gives a strongly consistent selection of the optimal number of change points, while an Akaike Information Criterion (AIC)-like criterion cannot. 
Finally, we demonstrate the theory and strength of the proposed algorithms by experiments on both synthetic and real-world data, including the Eastern US temperature data and the El Nino data from 1854 to 2015.
The experiment leads to some interesting discoveries about temporal variability of the summer-time temperature over the Eastern US, and about the most dominant factor of ocean influence on climate.
\end{abstract}

\begin{IEEEkeywords}
Autoregression, Information criteria, Strong Consistency, Stochastic Process.
\end{IEEEkeywords}

%
\IEEEpeerreviewmaketitle

%
%
%
%
\section{Introduction} \label{sec:intro}

\IEEEPARstart{S}{equentially} obtained data usually exhibits occasional changes in their structure, such as network anomalies in complex IP networks \cite{thottan2003anomaly}, distributional changes in  teletraffic models \cite{jana2000change},  sudden changes of volatility in stock markets due to financial crises \cite{hammoudeh2008sudden},   variations of an electroencephalogram (EEG) signal caused by mode changes in the brain \cite{vidal1977real}, or  environmental changes in various ecosystems \cite{hawkins2003detection,huntington2004matching}.
Change detection analysis tries to identify not only whether a time series is a concatenation of several segments, in which the neighboring ones are generated from different probability distributions, but also how many change points there are. 
%
%
There has been a vast amount of work in the filed of change point analysis.
In the parametric settings, the likelihood function naturally plays a key role, for example in the cumulative sum \cite{Basseville1993,Page1954} and the generalized likelihood ratio \cite{Gustafsson1996} approaches. Various tests have been developed for tracking changes in time series statistics  such as the mean \cite{Vogelsang1998}, the variance \cite{Inclan1994,Gombay1996}, the autocovariance function \cite{Berkes2009}, and the spectrum \cite{Picard1985}. 
Nonparametric approaches usually rely on kernel density estimation. 
A widely used approach
is to perform change detection by direct estimation of the ratio of probability densities \cite{Fishman1996,Huang2007,Sugiyama2008} or using some dissimilarity measure in feature space \cite{desobry2005online}, without estimating densities as an intermediary step. 
For  practical implementations, 
bisection procedure and its  extensions 
\cite{Vostrikova1981,scott1974cluster,Hawkins2001,Lavielle2006,fryzlewicz2014wild,cho2015multiple} have been widely studied.
Exact search methods such as segment neighborhood \cite{auger1989algorithms} and optimal partitioning \cite{yao1984estimation,jackson2005algorithm} have also been widely applied. 
Other remarkable progress in change point discovery for dependent time series data have been made in \cite{davis2006structural,khaleghi2012locating,khaleghi2014asymptotically,preuss2015detection}. 
More detailed references to the literature can be found in various remarkable monographs and review papers such as \cite{Basseville1993,Brodsky1993,Perron2006,jandhyala2013inference}.

As with any other statistical inference procedure, it is crucial to apply an appropriate model selection procedure in order to select the number of change points, whenever it is unknown. 
A common way is to apply the penalized approach, which selects the model dimension by minimizing the sum of goodness of fit and a penalty term. 
The three commonly used penalties are Akaike information criterion \cite{akaike1969fitting,akaike1998information},  Bayesian information criterion (BIC) \cite{schwarz1978estimating}, and Hannan and Quinn information criterion (HQ) \cite{hannan1979determination}. 
AIC is derived by minimizing the Kullback-Leibler divergence between the true distribution and the estimate of a candidate model,
BIC is from a large sample approximation that aims at selecting a model of maximum posterior probability,
HQ replaces the $\log N$ term in BIC by $c \log \log N (c>1)$, where $N$ is the sample size.
In some parametric models where regularity conditions are met, such as autoregressive models, it has been rigorously proved that AIC produces an overfitting model with non-vanishing probability, while BIC or HQ selects the model that converges almost surely to the true model (if it is included in the candidate set). In addition, HQ was proved to be the smallest penalty term that guarantees strong consistency, i.e, almost sure convergence \cite{hannan1979determination}.
Though these three criteria have been used as general-purpose model selection rules in various statistical models, their validity in terms of asymptotic behavior need to be verified case by case, especially for parametric or semi-parametric models where regularity conditions may not hold. Examples include finite mixture models \cite{keribin2000consistent,chen2001modified,chen2004testing,chen2008order,hui2015order},  and change point models considered in this paper. 
Despite remarkable progress on change detection analysis in defining the loss function and developing efficient algorithms, analysis on model selection 
	may benefit from further work.
To the best of our knowledge, only the consistency of BIC in selecting the number of change points for change detection problem have been studied in, e.g.,   \cite{yao1988estimating,venkatraman1992consistency,rigaill2015pruned}, but the theory of strong consistency for penalized method has not been well studied before.

A typical offline multiple change point analysis aims to solve the following problem. 
Given observations $y_1,\ldots,y_N$ $\in \mathbb{R}^D$ and $M \in \mathbb{N}$, 
the goal is to find integers  $0<\ell_1<\cdots<\ell_M<N $ that minimize the following sum of within-segment loss
\begin{align} \label{eq30}
e_M =\sum_{k=1}^{M+1} \textit{Loss}(y_{\ell_{k-1}+1},\ldots,y_{\ell_k}) , 
\end{align} 
where $\textit{Loss}(\cdot)$ is some selected loss function and by default $\ell_0=0,\ell_{M+1}=N$.
Here, specified number of change points $M$ is usually estimated using penalized approach.
A simple and widely adopted loss function is the quadratic loss \cite{rigaill2015pruned} defined by
	$	\textit{Loss}_q(x_{\ell_{j-1}+1},\ldots,x_{\ell_j}) = 
		\sum_{n=\ell_{j-1}+1}^{\ell_{j}} |x_n-\bar{x}|^2 $,	
	where $\bar{x}$ is the sample mean of $x_{\ell_{j-1}+1},\ldots,x_{\ell_j}$, and $|\cdot|$ denotes the Euclidean norm of a vector. 
One reason for using the quadratic loss is that it enables efficient $k$-means type fast implementations (discussed in Section~\ref{sec:finish}). 
Other commonly used loss functions include the negative log-likelihood associated with a specified parametric model \cite{horvath1993maximum,chen2011parametric}, or the cumulative sums \cite{page1954continuous,inclan1994use}.

In this work, we investigate the following two directions in  detecting structural changes in time series:   

1) In Section~\ref{sec:implementation}, we consider the formulation of change point analysis for a general stochastic process. 
	The basic idea is to assume that the time series data consists of several segments each of which is generated from a finite order autoregressive process. 
	For such dependent data, the loss function of each segment may be  defined as the log-likelihood loss associated with an autoregressive model, 
	and a standard change detection algorithm such as binary segmentation \cite{scott1974cluster,fryzlewicz2014wild} is amenable to use with the loss function.
	However, the loss function depends on a particular parametric assumption of the autoregression noises, and it does not always support efficient algorithms to minimize $e_M$. In fact, even if the noises are assumed to be Gaussian, the loss function can lead to massive computations, as we shall discuss it in detail later.
	To obtain the change points in a robust and computationally efficient manner, we propose an alternative approach which casts the change detection problem for the original time series $\{y_n\}$ into that for segment-wise (asymptotically) independent and identically distributed (i.i.d.) multivariate data $\{x_n\}$. We can discover the change points of independent data more easily, and then use the results to infer the change points of the original time series. 

2) In Section~\ref{sec:consistency}, we show that change points for a segment-wise independent data $\{x_n\}$ can be discovered by minimizing (\ref{eq30}) with  quadratic loss function and appropriately designed penalized methods.  
Specifically, we investigate necessary and sufficient conditions under which the unknown true number of change points can be determined for sufficiently large sample size (almost surely).

Finally, we present experimental results to demonstrate the performance of the proposed method on both synthetic and real-world datasets. In the study of real-world environmental data, we have used the Eastern US summer-time temperature data  from 1895 to 2015 and the El Nino data from 1854 to 2015. 
The experiments lead to  interesting conclusions about temporal variability of the summer-time temperature over the Eastern US, and about the most dominant factor of ocean influence on climate. These findings are consistent with those observed in the field of environmental sciences.



{\bf \textit{Notation and abbreviation}}:
%

Let $tr(\cdot)$, $(\cdot)^\T$, $\log$, $a.s.$, $i.o.$ respectively denote the trace of a square matrix, the transpose of a matrix or vector, natural logarithm, almost surely, and infinitely often.
We write $\mathcal{G} \sim [\mu,V]$ if distribution $\mathcal{G}$ has mean $\mu$ and variance $V$. 
We say ``$h(N) > 0$ tends to infinity as $N$ tends to infinity'' if $\lim_{N\rightarrow \infty}1/h(N)=0$.
We write $h(N)=\Theta(g(N))$ if $c <  h(N)/g(N) < 1/c$ for some constant $c \neq 0$ for all sufficiently large $N$.
We write $h(N)=o(g(N))$  if $\lim_{N \rightarrow \infty} h(N)/ g(N) = 0 $.
Let $\mathcal{N}(\mu,V)$ denote the multivariate normal distribution with mean $\mu$ and covariance matrix $V$.
Let $C$ denote a generic constant. We use $o_p(1)$ and $O_p(1)$ to respectively denote any random variable that converges in probability to zero and that is stochastically bounded.
Throughout the paper, random variables and observed data are respectively represented by capital letters (e.g. $Y_n$) and small letters (e.g. $y_n$). Vectors are  all column vectors. 

A generic change detection model assumes data to be the outcomes of a sequence of multi-dimensional real-valued random variables $\{Y_n:n=1,\ldots,N\}$
that consists of of $M_0+1$ segments ($M_0 \in \mathbb{N} \cup \{0\}$), where each pair of neighboring segments have different data generating process.
In this paper,  $Y_n$'s are sometimes substituted with $X_n$'s in order to emphasize the independence of data.  
We denote the true segments by $\{Y_n:n=L_{k-1}+1, \ldots,L_{k}\}$, $k=1,\ldots,M_0+1$, where $L_{1}<\cdots<L_{M_0}$ are referred to as the $M_0$ change points, and by default $L_{0}=0,L_{M_0+1}=N$.  
Let $N_k=L_{k}-L_{k-1},k=1,\ldots,M_0+1$ denote the size (length) of the $k$th segment. Clearly, $\sum_{k=1}^{M_0+1}N_k = N$.
Throughout the paper, we use $\hat{M}$ to denote the estimated  number of change points. 
Similarly, we represent the detected change points by $\hat{L}_{k},k=0,\ldots,\hat{M}+1$, and segment sizes by $\hat{N}_k,k=1,\ldots,\hat{M}+1$.

\section{Change Detection for Time Series with Dependency Structure} \label{sec:implementation}



In this section, we consider a sequence of one-dimensional dependent data.
The results can be readily  extended to multi-dimensional data.
We assume that the data is generated from the following \textit{segment-wise  autoregressive (AR) model}:

(M.1) The sequence $\{Y_n:n=1,\ldots,N \}$ are one-dimensional and it consists of $M_0+1$ segments, each of which can be described by a linear autoregression.
In other words, for each $k=1,\ldots,M_0+1$, we have
$
Y_n=  \underline{Y}_n^\T  \psi^{(k)} + \v^{(k)}_n, \, n=L_{k-1}+1,\ldots,L_{k}, 
$
where 
$ \underline{Y}_n=[1, Y_{n-1},\ldots,Y_{n-L}]^\T$ (for $L>0$) or $\underline{Y}_n= 1 $ (for $L=0$), $ \psi^{(k)} \in \mathbb{R}^{L+1}$
 (referred to as AR filter of order $L$), 
	 $Y_{1-L},\ldots,Y_0$ have been used to denote initial values, and $\v^{(k)}_n$ are independent noises and are i.i.d. within each segment.
Moreover, $ \psi^{(k)} \neq  \psi^{(k+1)}, k=1,\ldots,M_0$.

An autoregression of order $L \in \mathbb{N} \cup \{0\}$ is also denoted by AR($L$).
Note that we have assumed the same $L$ in each segment of the data generating model for the simplicity of presentation (clearly, any AR($\ell$) is necessarily AR($k$) for $\ell < k < \infty$, so that we may let $L$ be the maximum of all the AR orders from each segment). 
In the rest part of the paper, we assume that the order $L$ is known as prior knowledge or from exploratory studies.
Our goal is to identify the number of change points and their locations. 

Before we proceed, it is worth mentioning that even though the above change detection model is semi-parametric since no assumption on how each AR model switches to another one was made, the change point analysis can serve as an exploratory study for more parametric settings. For example, the detected change points can be used to set up better initial values of Expectation-Maximization algorithm for complex parametric mixture models such as point process regression models \cite{sheikhattar2015recursive} and multi-state autoregressive models \cite{ding2015learning}.  

It is natural to define the loss function based on 
\begin{align} \label{eq95}
	\textit{Loss}_{a}(y_{\ell_{j-1}+1},\ldots,y_{\ell_j}) = 
	\sum\limits_{n=\ell_{j-1}+1}^{\ell_{j}} (y_n- \underline{y}_n^\T \hat{ \psi} )^2		
\end{align}
where $\hat{\psi}$ is estimated from $y_{\ell_{j-1}+1},\ldots,y_{\ell_j}$ by Yule-Walker equation or least squares method. 
The above loss is interpreted as the cumulated prediction errors, or the rescaled negative log-likelihood associated with AR($L$). 
The quadratic loss can be regarded as the special case when $L=0$.
We can find change points by minimizing the sum of within-segment loss in (\ref{eq30}) using state-of-the-art algorithms such as  binary segmentation \cite{scott1974cluster}. 
However, an alternative idea is to turn the change detection of segment-wise autoregressive model into that of segment-wise Gaussian independent model.

{\bf Discussion  of our motivation}:
Consider a sequence of  $N$ points that are generated from a single AR($L$), i.e. $Y_n =  \psi^\T  \underline{Y}_n + \v_n$,
where 
$ \psi \in \mathbb{R}^{L+1}$,
$\v_n \sim [0,\sigma^2]$. 
Suppose that the true change points of $\{Y_n: n=1,\ldots,N\}$ are located at multiples of $w$, where $w>2L$ is an integer, and the data are divided into $N/w$ segments of size $w$. If each segment of data is used to estimate an  AR($L$) filter, we obtain $N/w$ estimates of $ \psi$, respectively denoted by $ \hat{ \psi}_{1},\ldots, \hat{ \psi}_{N/w}$. 
It has been well established that if $\hat{ \psi}_{i}$  is estimated from either least squares or Yule Walker methods, 
$\sqrt{w} (\hat{ \psi}_{i} -  \psi)$ converges in distribution to $\mathcal{N}(0,\Gamma) $ as $w$ goes to infinity, where $\Gamma$ is a constant matrix depending only on $ \psi$ \cite[Appendix 7.5]{box2011time}.
Thus, $\hat{ \psi}_{i}$ can be approximated by multivariate Gaussian random variables with mean $ \psi$ and variance  $\Gamma/w$. 
The asymptotic independence of $\sqrt{w} (\hat{ \psi}_{i} -  \psi)$ are guaranteed by the following result.
\vspace{0.1cm}	 
\begin{theorem} \label{thm:independency}
	Suppose that $\{Y_1,\ldots,Y_N\} $ are generated from an autoregression with filter $ \psi$. 
	Let $\hat{ \psi}_1 \in \mathcal{R}^{L+1}$ and $\hat{ \psi}_2\in \mathcal{R}^{L+1}$ respectively denote the estimated filters from $\{Y_1,\ldots,Y_{N_1}\}$ and $\{Y_{N_1+1},\ldots,Y_N\}$ by least square methods,
	where $N_1, \, N_2= N-N_1 \rightarrow \infty$ as $N \rightarrow \infty$.
	Assume that 
	
	(A.1) the noises $\v_n$ satisfy $E[ \max\{(\log |\v_n|), 0\}] < \infty$ and  the distribution of $\v_n$ has  nontrivial absolutely continuous components.
	
	Then $\sqrt{N_1}(\hat{ \psi_1}- \psi)$ and $\sqrt{N_2}(\hat{ \psi_2}- \psi)$ converge to two Gaussian random variables that are independent. 
\end{theorem}
\vspace{0.1cm}	

Assumption (A.1) is mild, as for instance, it is satisfied by the Gaussian distribution.
Theorem~\ref{thm:independency} implies that if a data from the same autoregression is split into two (or more) parts, and each part gives an estimate of the true filter, then the estimators are asymptotically independent (up to a rescaling). 

Now suppose that the stochastic process consists of two parts: the first $N_1$ points are generated from one AR($L$)  and the rest $N_2$ are from another AR($L$). If a window size $w$ that satisfies $2L < w < \min\{N_1,N_2\}$ is chosen, the estimated AR filters are approximately independent points in $\mathbb{R}^{L+1}$ and they contain a change point around the ($N_1/w$)th point. Here and afterwards, we assume that $N_1/w,N_2/w$ are integers. Extension to more general cases is straightforward.  
We propose a multi-window (MW) change detection algorithm that chooses different $w$'s and collect the information of the detected change points for each $w$ in a proper way, in order to obtain a more accurate estimation of the change points of the stochastic process. From a computational point of view, starting from a large $w$ also helps to reduce the cost, which is especially helpful in cases where massive time series data is involved. 

Algorithm~\ref{algo:multiWindow} is a pseudo-code for MW method, followed by two subroutines: Algorithms~\ref{algo:oracle} and \ref{algo:peak}. 
Illustrating experiments are provided in Section~\ref{sec:experiments}.
Algorithm~\ref{algo:multiWindow} uses a sequence of $R$ window sizes $w_1>\cdots>w_R$ (discussed below) in order to capture any true segment of small size. For each $w_r$, the original data is turned into a sequence of $L+1$ dimensional data that can be approximated as independent. 
By calling Algorithm~\ref{algo:oracle}, we obtain a set of change points $\hat{\ell}_1,\ldots,\hat{\ell}_M$; 
By further mapping these change points back to the original scale $\{1,\ldots,N\}$, we obtain several short ranges (intervals) $I^{(r)}_k$ (each of size $2w_r$) that ``probably'' contain the desired change points. 
We repeat the above procedure for different $w_r$, and combine the information in the following way: 
the detected ranges of change points from each window size are scored by one, the scores are aggregated, and the ranges with highest score or around the highest score (determined by the tolerance parameter $\tau$) are finally selected. 
The output of the algorithm is  $\hat{M}$ number of ``peak'' ranges that are most likely to  contain the true change points.

\vspace{-0.2cm}%

\begin{algorithm}[H]
\vspace{-0.0 cm}
\small
\caption{change detection by multi-window  method}
\label{algo:multiWindow}
\begin{algorithmic}[1]
\INPUT $\{ y_n \in \mathbb{R}, n=1,\ldots , N\}$, $L$ (lag order), $M_{\text{max}}$ (the largest size of candidate models), 
	$w_{1}>\cdots>w_R$ (window sizes)
\OUTPUT  $\widehat{cp}=\{ \hat{I}_1,\ldots,\hat{I}_{\hat{M}} \}$ (ranges containing change points)
\STATE $s^{(0)}_n=0, n=1,\ldots,N$ (initialized score)
\FOR {$r = 1 \to R$ } 
	\STATE Let $N_{r} = N/w_r$. Estimate $\hat{ \psi}_{n_{r}} \in \mathbb{R}^{L+1} $ from $\{Y_{n}:n=(n_r-1)w_r+1,\ldots,n_r w_r\}$, $n_r=1,\ldots,N_r$.
   	 \STATE Call Algorithm~\ref{algo:oracle} with input $\hat{ \psi}_{n_r}:n_r=1,\ldots,N_r$, $M_{\text{max}}$, selected $f(N)$, $\beta(N)$, 
   	 		and obtain output $\hat{\ell}_1,\ldots,\hat{\ell}_{M_r}$
     \STATE Define scores $s^{(r)}_n = s^{(r-1)}_n+{\bf 1}_{n \in \bigcup_{k=1}^{M_r} I^{(r)}_k}$, $n=1,\ldots,N$, where $
     	I^{(r)}_k =   [(\hat{\ell}_{k}-1) w_r+1, (\hat{\ell}_{k}-1) w_r+2, \ldots, (\hat{\ell}_k+1) w_r ] 
     $, 
     and 
     ${\bf 1}_{n \in A}$ equals one if $n$ belongs to the set $A$ and zero otherwise.
\ENDFOR
\STATE Call Algorithm~\ref{algo:peak} to obtain the peak ranges $\widehat{cp}=\{ \hat{I}_1,\ldots,\hat{I}_{\hat{M}} \}$ ($\hat{M} \leq M_{\text{max}}$), whose associated scores are at least $S- \tau$ where $S=\max_{n=1,\ldots,N}\{s^{(R)}_n\}$ is the highest score 
\end{algorithmic}
  \vspace{-0.0cm}%
\end{algorithm}

\vspace{-0.2cm}%

\vspace{-0.2cm}%

\begin{algorithm}[H]
\vspace{0.0 cm}
\small
\caption{(generic) change detection by minimizing the sum of within-segment quadratic loss}
\label{algo:oracle}
\begin{algorithmic}[1]
\INPUT $\{ x_n \in \mathbb{R}^D, n=1,\ldots , N\}$, $M=M_{\text{max}} \in \mathbb{N}$ (the largest candidate number of change points), $f(N)$ (penalty term), $\beta(N)$ (minimal segment size)
\OUTPUT  $\hat{M}$, $\hat{\ell}_1,\ldots,\hat{\ell}_{\hat{M}}$ (discovered change points).
    \FOR {$k = 0 \to M_{\textrm{max}} $ }
   	 \STATE Define $\ell_0=0,\ell_{k+1}=N$; minimize  
    	$ 
    	e_k=\sum_{j=1}^{k+1}\textit{Loss}_q(x_{\ell_{j-1}+1},\ldots,x_{\ell_j})
		$ 
		over $\ell_j \in \mathbb{N}$ and record the optimum $\hat{\ell}_j : j=1,\ldots,k , \ \hat{e}_k$
		\IF { size of the smallest segment $<$ $\beta(N)$}
     		\STATE Let $M=k-1$; break the for loop
     	\ENDIF
     \ENDFOR
\STATE Choose 
		$\hat{M} = \argmin_{k=0,\ldots,M} 
			(\hat{e}_k + k f(N))$,
	and $\{\hat{\ell}_j\}$ to be the solution to Step 2 
	under $k=\hat{M}$.
    
\end{algorithmic}
  \vspace{0.0cm}%
\end{algorithm}

\vspace{-0.2cm}%



\vspace{-0.2cm}%
\begin{algorithm}[H]
\vspace{0.0 cm}
\small
\caption{peak range selection}
\label{algo:peak}
\begin{algorithmic}[1]
\INPUT $s_n^{(r)}, n=1,\ldots,N, r = 1,\ldots,R$ (recorded scores), $\tau \in \mathbb{N} \cup \{0\}$ (tolerance level)
\OUTPUT  $\widehat{cp}=\{ \hat{I}_1,\ldots,\hat{I}_{\hat{M}} \}$ (the output of Algorithm~\ref{algo:multiWindow})
\FOR {$r = R \to 1$ } 
	\STATE Let $cp=\cup_{m=1}^{M^*_r} J_m $ be the union of all the peak ranges whose associated scores are no less than $S-\tau$. In other words, there exist positive integers $u_{m},v_{m}$, $m=1,\ldots,M^*_r$ such that 
	$u_{1} < v_{1} < u_{2} <  \cdots  < v_M$, 
	$J_m=[u_{m}+1,u_{m}+2 , \ldots,v_{m}]$,  
	$s^{(r)}_{u_{m}+1}=\cdots=s^{(r)}_{v_{m}}  \geq S- \tau$ for each $m=1,\ldots,M^*_r$, 
	and $s^{(r)}_n < S- \tau $ for all $n \not\in cp$. (Note that the above conditions have been designed to obtain peak ranges as narrow as possible.) 
	\IF { $M^*_r \leq M_{\text{max}}$} 
		\STATE Let $\widehat{cp}=cp$, namely $\hat{M} = M^*_r$ and $\hat{I}_m = J_m$ for each $m= 1, \ldots, \hat{M}$; 
		 break the for loop
	\ENDIF 
\ENDFOR
\end{algorithmic}
  \vspace{0.0cm}%
\end{algorithm}
\vspace{-0.2cm}%

Some detailed discussions of Algorithm~\ref{algo:multiWindow} are given below.

{\bf \textit{Algorithm~\ref{algo:oracle}}:} This subroutine is called by Algorithm~\ref{algo:multiWindow} at Step 4. It detects the number and locations of change points based on minimizing within-segment quadratic loss and applying penalized model selection approach (with $D=L+1$). For clarity, we focus on Algorithm~\ref{algo:multiWindow} in this section, and defer detailed discussions of Algorithm~\ref{algo:oracle} to Section~\ref{sec:consistency}. 

{\bf \textit{Algorithm~\ref{algo:peak}}:} This subroutine is called by Algorithm~\ref{algo:multiWindow} at Step 7. It aims to selects the narrow ranges with score at least $S- \tau$, which are most likely to contain change points. 
	Its ``for'' loop (from $R$ to $1$) 
	is a backward pruning procedure in order to ensure $\hat{M} \leq \m$. 
	The pruning was done by neglecting scores produced by the smallest window sizes, which are less reliable as the estimated AR filters from those windows have larger variances.

{\bf \textit{Window sizes}:}
Intuitively speaking,  more reliable change detection results can be obtained by using multiple window sizes (instead of only one), since in practice we do not know what the true segment sizes are, and an inappropriately chosen $w_r$ may be so large that a true segment is ``missed''.
On the other hand, a small $w_r$ leads to larger variance of AR filter estimates.
A properly designed MW method strikes a tradeoff between estimation accuracy (since larger window sizes reduce variance of the estimated AR filters) and the resolution of the detected change points (since smaller window sizes produce narrower ranges).

{\bf \textit{Computing the estimator $\hat{ \psi}_{n_{r}}$}:}
For a specified AR order $L$, $\hat{ \psi}_{n_{r}}$ can be obtained either by least squares method, or by the Yule-Walker method (which requires slightly more data points, but supports fast computation by, e.g., the Levinson-Durbin recursion~\cite{brockwell2013time}).  

{\bf \textit{Tolerance parameter}:} The main purpose of introducing the tolerance parameter $\tau$ in step 8 of Algorithm~\ref{algo:multiWindow} is to ensure that the scoring produces fair comparisons among different ranges. Otherwise, small segments may be  ``missed'' by some initial large window sizes. For example, suppose that $\tau=0$, $w_1=200$ and there is only one true change point at $N_1=50$ in $N=1000$ data points. Then in this scenario, it is harder to discover a change point from $N/w_1=5$ estimated filters.



It is worth mentioning that the output of MW method is a set of $\hat{M}$ narrow ranges instead of single points. In the cases where $\hat{M}$ exact change points are desired, we  could use the results from Algorithm~\ref{algo:multiWindow} as starting point to further search optimal points within those ranges. 
In that sense, MW method can serve as a fast prescreening approach. 
In addition, the multiple windows can be implemented in parallel for massive time series, and it can be applied to independent data as well. 

\section{Strong Consistency of Penalized Methods} \label{sec:consistency}

In this section, 
we discuss subroutine Algorithm~\ref{algo:oracle}. 
This subroutine discovers change points by minimizing the within-segment sum of quadratic loss $\hat{e}_{k}$. 
We will show that when applied to a segment-wise independent data, Algorithm~\ref{algo:oracle} outputs $\hat{M}$ such that $\hat{M} \overset{a.s.}\longrightarrow M_0$ as data size tends to infinity. 
  
Algorithm~\ref{algo:oracle} computes $\hat{e}_{k}$ for each candidate number of change points $k=\{0,\ldots, M\}$, where $M$ is determined by the largest candidate number of segments $\m$ and minimal segment length $\beta(N)$.
After that, the optimal number of change points is estimated according to a penalized method. 
Further details are given below. 

{\bf \textit{Parameter $\beta(N)$}:} It is introduced for two purposes: for the technical convenience in deriving asymptotic results, and for faster implementation in practice. 
$\beta(N)$ must be selected such as $\lim_{N \rightarrow \infty}\beta(N) = \infty$.
The rate of growth of $\beta(N)$ will be selected depending on the theoretical results we wish to prove, as we shall discuss this later. 

{\bf \textit{Penalty function}:}
The common choice of penalty function is a linear function in the form of 
$k f(N)$, where $f(N)$ is referred to as the penalty term. 
For brevity, we consider the linear function in this paper, but the results can be applied to more general penalty functions. 
Three commonly used types of penalty terms are related to AIC, HQ, and BIC.
In a parametric change detection problem, if there are $k$ change points and $p$ parameters in each segment, the total number of parameters to appear in AIC and BIC is $k+p(k+1)$. If the quadratic loss is treated as twice the negative log-likelihood of a Gaussian probability density function with variance equal to the identity matrix, the total number of parameters is $k+D(k+1) = k(D+1)+constant$.
The penalty terms $f(N) \propto 1$, $f(N) \propto \log\log N$, and $f(N) \propto \log N$ are referred to as the
 variants of AIC, HQ, and BIC, respectively. 

{\bf \textit{Strong consistency}:}
A penalized model selection approach is referred to be strongly consistent if $\hat{M} \overset{a.s.}\longrightarrow M_0$ as data size tends to infinity. We may also say that $\hat{M}$ is strongly consistent. 

We make the following assumption about a segment-wise independent time series.	

\vspace{0.1cm}

(M.2) 
	The sequence $\{X_n:n=1,\ldots,N\}$ are $D$-dimensional ($D \in \mathbb{N}$) and independent random variables. Moreover, for each $k=1,\ldots,M_0+1$, we have $\lim_{N \rightarrow \infty}N_k = \infty$,  
	and $\{X_n: n=L_{k-1}+1, \ldots,L_{k}\}$ are i.i.d. distributed according to $\mathcal{G}_k\sim [\mu_k, V_k]$. 
	When $M_0 \geq 1$, 
	$\mu_{k} \neq \mu_{k+1}$, $k=1,\ldots,M_0$.

\subsection{Necessary conditions for strongly consistent model selection} \label{sec:necessary}

We start by examining the case when the true data generating process has no change point.

%

\vspace{0.1cm}
\begin{theorem} \label{thm:loglog0changepoint}
	Assume that the data generating model is given by (M.2) with $M_0=0$. 
	Then the smallest penalty term $f(N)$ that guarantees strong consistency of $\hat{M}$ in Algorithm~\ref{algo:oracle} is at least $\Theta(\log\log N)$. 
	If we additionally assume $\beta(N) = \Theta (N)$, then there exists a constant $C>0$ such that $f(N) = C \log\log N$ 
	guarantees strong consistency of $\hat{M}$.
\end{theorem}
\vspace{0.1cm}


Theorem~\ref{thm:loglog0changepoint} implies that the smallest penalty for strong consistency is  $\Theta(\log \log N)$ (given by variants of HQ criterion).
A by-product of its proof is a technical lemma (Lemma~\ref{lemma_2} in the appendix) 
that implies that an AIC-like criterion (with constant penalty) always produces a non-vanishing overfitting probability. We recap this observation after Lemma~\ref{lemma_2}. Interestingly, these observations are similar to those found for order selection of autoregressive models, even though an autoregressive model is purely parametric, and the proof in those cases require different technical approaches \cite{hannan1979determination,shibata1976selection}.  

Next, we consider $f(N)=\Theta( \log\log N )$ for the case $M_0>0$. 
We define 
 \begin{align} \label{eq31}
	\bar{\Delta}_{\mu} = \max_{k=1,\ldots,M_0}\{|\mu_k-\mu_{k+1}|\} , \,
	\underline{\Delta}_{\mu} = \min_{k=1,\ldots,M_0}\{|\mu_k-\mu_{k+1}|\}  
 \end{align}



	
\begin{theorem} \label{thm:loglog}
	Under the model assumption (M.2) with $M_0>0$, suppose that $\beta(N)=\Theta(N)$ and 
	
	(A.2) The largest candidate number of change points $M_{\text{max}}$ is finite and $M_{\text{max}} \geq M_0+3$,
	
	(A.3) The true segment sizes satisfy $\beta(N)\leq N_k/4$, $N_k=\Theta(N)$ for $k=1,\ldots,M_0+1$. In addition, $f(N) = o(N)$.
	
	Then there exists a positive constant $C_0$ such that whenever $f(N) \geq C_0 \log\log N$,  the estimated number $\hat{M}$ satisfies $M_0 \leq \hat{M} \leq 2M_0$ for sufficiently large $N$ almost surely, namely
	$$
	\P \biggl\{ \limsup\limits_{N \rightarrow \infty} 
	(\hat{M} < M_0) \cup (\hat{M} > 2M_0) \biggr\}=0.
	$$
	Moreover, the distances between the estimated change points and true ones satisfy 
	\begin{align} \label{eq91}
		 \limsup\limits_{N\rightarrow \infty}\min_{k=1,\ldots,\hat{M}} \frac{|\hat{L}_{k}-L_{j}| }{ 2\beta(N)}  \leq 1 \quad (a.s.) 
	\end{align}
	for each $j=1,\ldots,M_0$. 
\end{theorem} 
\begin{remark}
The requirement $M_{\text{max}} \geq M_0+3$ (instead of $M_{\text{max}} \geq M_0$) in (A.2) is for technical convenience in the proof of Theorem~\ref{thm:loglog}. 
	Theorem~\ref{thm:loglog} shows that   $f(N)=\Theta(\log \log N)$ suffices to guarantee no underfitting. Although we cannot prove it avoids overfitting as well, we proved that the extent of overfitting is bounded (since $\hat{M} \leq 2M_0$ holds almost surely).  In addition, Inequality (\ref{eq91}) implies that each true change point is ``almost'' captured, since its nearest discovered change point is within distance  $\beta(N)$, which can be chosen to be arbitrarily small compared with $N$ (or each $N_j$). 
	In the next subsection, we relax the assumption on $\beta(N)$ and obtain strongly consistent $\hat{M}$ by increasing the penalty to be BIC-like. 
\end{remark}

\subsection{Sufficient conditions for strongly consistent model selection} \label{sec:sufficient}

\textit{Definition}: A real-valued random variable $X$ is said to be sub-Gaussian if it has the property that there exists a constant $b>0$ such that for every $t \in \mathbb{R}$, one has 
$E(e^{t\{X-E(X)\}}) \leq e^{b^2t^2/2}$. 
It is easy to prove using Markov inequality that there is some $c_0 > 0$ such that for every $a \in \mathbb{R}$,
\begin{align} \label{eq92}
	\P( |\bar{X}-E(X)| \geq a ) \leq 2e^{-c_0 a^2 N}
\end{align}
where $\bar{X}$ is the mean of i.i.d. random variables $\{X_n:n=1,\ldots,N\}$.
Assuming that $X_n$ follows a sub-Gaussian distribution,  it is possible to prove the strong consistency of $\hat{M}$. Intuitively speaking, it is because of the quick decay of tail probability densities that enables union bounds for the behavior of discovered change points.   


\begin{theorem} \label{thm:strongConsistency}
	Under the model assumption (M.2) with $M_0 \geq 0$, suppose that Assumptions~(A.2), (A.3) in Theorem~\ref{thm:loglog} hold and that 
	
	(A.4) $\mathcal{G}_k , k=1,\ldots,M_0+1$ are marginally sub-Gaussian. In other words, there exists a constant $c_0 > 0$ such that (\ref{eq92}) holds for each marginal distribution of $\mathcal{G}_k$.
	
	If 
	\begin{align}
		f(N) \geq 100 \bar{\Delta}_{\mu}^2 \eta^{*}(N),  
		\label{condition1}
	\end{align} 
	where 
	$$
	\eta^{*}(N) = \frac{250 D  c^{M_0-1}}{ c_0 \underline{\Delta}_{\mu}^2} \log N, \quad
		c=4/(\sqrt{2}-1)^{2},
	$$ 
	then	 
	$\hat{M}$ is strongly consistent. 
	Moreover,
	$$
	\limsup\limits_{N \rightarrow \infty} \biggl( \max_{k=1,\ldots,M_0} \frac{|\hat{L}_{k} - L_{k}|}{\eta^{*}(N)} \biggr) \leq 1  \quad (a.s.) 
	$$
\end{theorem}

\begin{remark} \label{understand_condition}
	Assumption~(A.4) is satisfied by Gaussian, any bounded random variables, etc. 
	By the conditions of Theorem~\ref{thm:strongConsistency}, both the minimal distance and the minimal penalty required for strong consistency are no more than $\Theta(\log N)$. Note that we do not need the requirement $\beta(N)=\Theta(N)$.
	The constant term for $f(N)$ is proportional to the dimension $D$ and 
	 the ratio $\bar{\Delta}_{\mu}^2/\underline{\Delta}_{\mu}^2$. 
	Intuitively,  higher dimension and larger variance require stronger penalties. Besides this, it is interesting to observe that $f(N)$ depends on the ratio $\bar{\Delta}_{\mu}^2/\underline{\Delta}_{\mu}^2$ which is scale invariant, while $\eta^{*}(N)$ only depends on the smallest distance between two neighboring distributions (in terms of the means). 
\end{remark}

\subsection{Implementation of Algorithm~\ref{algo:oracle}} \label{sec:finish}

Implementations of Algorithm~\ref{algo:oracle} based on popular methods such as binary segmentation \cite{scott1974cluster}, segment neighborhood \cite{auger1989algorithms}, and optimal partitioning \cite{yao1984estimation,jackson2005algorithm} are possible.
But since our loss function is quadratic, it is possible to have an algorithm that takes full advantage of this fact. We propose such a computationally efficient algorithm, which 
is analogous to but also differs from the usual k-means algorithm (in that each segment/cluster contains points with consecutive indices). 
It can then be regarded as an ``ordered k-means'' algorithm.
The algorithm reduces the within-segment quadratic loss in each step by moving the change points based on the following result.
\begin{proposition} \label{lemma_new1}
	Suppose that $\{X_{n}:n=1,\ldots,N_1\}$	and $\{X_{n}:n=N_1+1,\ldots,N_1+N_2\}$ are two segments.
	Consider the operation that shifts the change point from $N_1$ to $N_1-t$ where $0<t<N_1$: the two segments become $\{X_{n}:n=1,\ldots,N_1-t\}$	and $\{X_{n}:n=N_1-t+1,\ldots,N_1+N_2\}$. 
	The within-segment quadratic loss will be reduced after the operation if  and only if 
	\begin{align*} 
		\frac{N_1 |\bar{X}_{0,N_1}-\bar{X}_{N_1-t,N_1}|^2}{N_1-t}	>
		\frac{N_2|\bar{X}_{N_1,N_1+N_2}-\bar{X}_{N_1-t,N_1}|^2}{N_2+t} ,
	\end{align*}
 where $\bar{X}_{n_1,n_2}$ denotes the sample mean of $\{X_n:n=n_1+1,\ldots,n_2 \}$.
\end{proposition}
From the above lemma, it follows that in order to decide whether a subsequence of data should be moved from one segment to its neighboring one, it only suffices to compute its mean and also the means of the original two segments. 
By iterative application of the above lemma, a local optimum of 
step 2 in Algorithm~\ref{algo:oracle} could be achieved.

\section{Experiments} \label{sec:experiments}

In this section, we present experimental results to demonstrate the above theoretical results, and the advantages of MW method on both synthetic and real-world datasets.
The algorithms were implemented in Matlab and run on a PC with 3.1 GHz dual-core CPU. The source codes and related data will be made public online in the future. 
In the experiments, we rescale the penalty term $kf(N)$ in Algorithm~\ref{algo:oracle} to $var(X)kf(N)$. Although it does not affect our theoretical results, 
it is convenient for implementation in practice. 

\subsection{Independent data}
In a synthetic data experiment, we generated data of two change points: $X_n \sim \mathcal{N}(\mu_1,\sigma^2)$, $n = 1,\ldots,0.2N$, $X_n \sim \mathcal{N}(\mu_2,\sigma^2)$, $n = 0.2N+1,\ldots,0.8N$, $X_n \sim \mathcal{N}(\mu_3,\sigma^2)$, $n = 0.8N+1,\ldots,N$. 
Let $[\mu_1,\mu_2,\mu_3,\sigma^2]=[-1,0,1,1]$, $M_{\text{max}}=10$, $f(N)=2 \log N$, $\beta(N)=\log\log N$.
For illustration purpose, an example dataset with $N=100$ is plotted in Fig.~\ref{fig:synthetic:numcp}(a).
For each $N=100,500,1000$, we generate 100 independent datasets and summarize the detected change points (normalized by $N$) in Fig.~\ref{fig:synthetic:numcp}(b). We also summarized the percentage frequencies of $\hat{M}<2$, $\hat{M}=2$, and $\hat{M}>2$, respectively denoted by $f=(f_1,f_2,f_3)$. They are $f=(38,60,2)$ for $N=100$, $f=(0,89,11)$ for $N=300$, and $f=(0,95,5)$ for $N=1000$.  
The results show that both the estimated number of and locations of change points become more and more accurate as the sample size grows. 

\begin{figure}
\centering
  \includegraphics[width=0.6\linewidth]{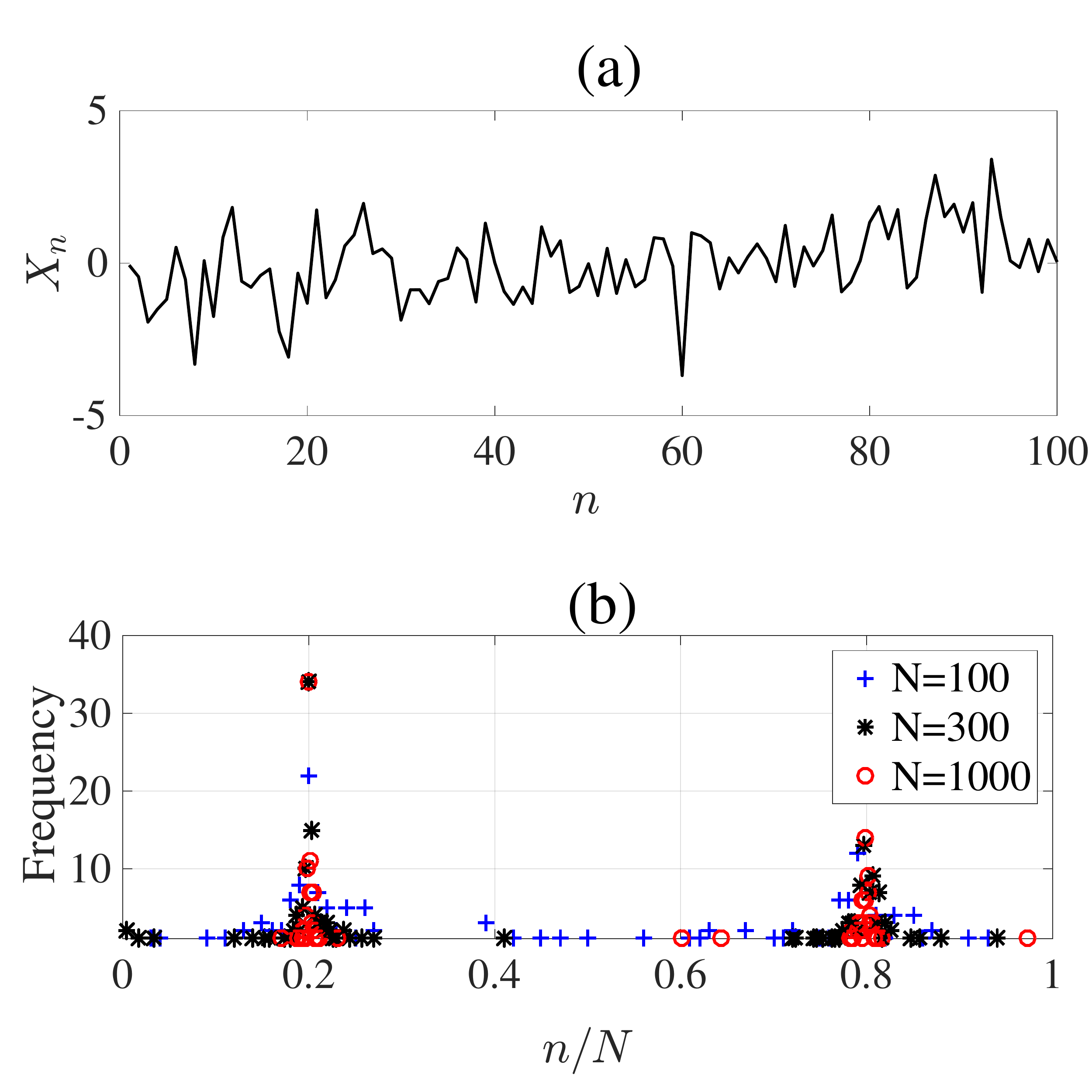}
  \vspace{-0.2in}
  \caption{(a) A sequence of independent data that contains two change points, and (b) the frequencies of discovered change points for each $N=100,300,1000$}
  \label{fig:synthetic:numcp}
  \vspace{-0.2in}
\end{figure}


\subsection{Dependent data}

In a synthetic data experiment for dependent data, we generated data of two change points at $0.1N$ and $0.3N$. 
Data is generated from a zero mean autoregression in each of the three segments, and the associated AR filters  are respectively $[\psi^{(1)}_1,\psi^{(1)}_2]= [0.8,-0.3]$, $[\psi^{(2)}_1,\psi^{(2)}_2]=[-0.5,0.1]$,$[\psi^{(3)}_1,\psi^{(3)}_2]=[0.5,-0.5]$.
Suppose that the noises are $\mathcal{N}(0,1)$ and $M_{\text{max}}=5$, $f(N)= \log N$, $\tau=1$. 
Fig.~\ref{fig:synthetic:ARDataExample}(a) illustrates one dataset with $N=1000$. 
We set window sizes to be $[w_1,w_2,w_3,w_4]=[100 ,50, 20, 10]$ and apply Algorithm~\ref{algo:multiWindow} to that dataset.
The score is plotted in Fig.~\ref{fig:synthetic:ARDataExample}(b).

\begin{figure}
\centering
  \includegraphics[width=0.6\linewidth]{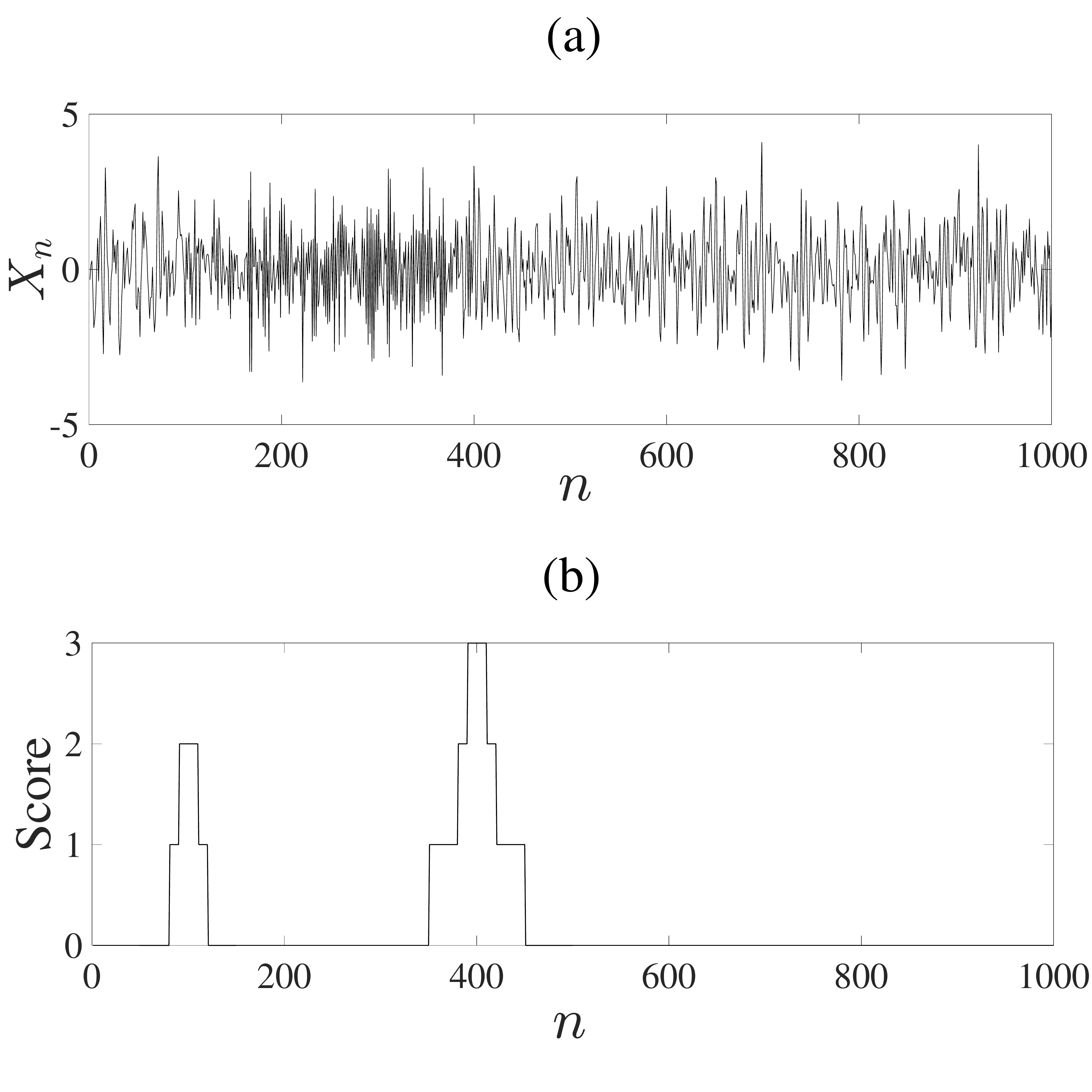}
  \vspace{-0.2in}
  \caption{(a) A time series that consists of three segments of various autoregressions, and (b)  score plot for change detection}
  \label{fig:synthetic:ARDataExample}
  \vspace{-0.2in}
\end{figure}

Next, we compare MW method with binary segmentation (BS) method (which is perhaps the most widely applied approach in the literature).
The BS method first scans all the points and finds a single change point that minimizes the sum of within-segment loss defined in (\ref{eq95}), and then extends to multiple change points discovery  by iteratively repeating the method on different subsets of the series. 
	This procedure is repeated until the maximal number of change points is reached or no more change point is detected.
By assuming that $L$ is a constant, the complexity of BS algorithm for segment-wise AR of size $N$ is calculated to be in the order of $\Theta(N^2)$, while MW method is of $\Theta(N+N^2/w_R^2)$ ($w_R$ is the smallest window size).
To compare the performance of MW and BS, we repeat the above experiment for 50 iterations. In each iteration, we generated three autoregressive filters of order $L=2$ that are independent and uniformly distributed in the space of all stable AR(2) filters.\footnote{In general, for a stationary AR($L$) processes with coefficients $ \psi = [\psi_1,\ldots,\psi_L] $, $ \psi$ stays in a bounded subspace $S_L \subset \mathbb{R}^L$. 
	For the purpose of fair comparison, in the experiment we draw AR filters that are uniformly distributed on $S_L$, using the technique proposed in \cite{ding2015data}.}
The change points are still $0.1N$ and $0.3N$. The number of points is $N=10^4$. 
The discovered change points are plotted in Fig.~\ref{fig:compare}(a).
In order to compare the computational speed,
we repeat the above experiment for each $N = [10^3,5\times 10^3,10^4,5\times 10^4, 10^5]$.  
For the MW method, we use fixed number of windows $\{w_r\}_{r=1}^4 = N/10, N/20, N/50, N/100$ and tolerance parameter $\tau = 2$. 
	We set the minimal length for BS method to be $10L$ (which is used to guarantee stability involved in matrix computations). 
	For both methods, $\m=4$. The comparison is plotted in Fig.~\ref{fig:compare}(b).
The average numbers of detected change points (with standard deviation inside the parenthesis) under each $N$ are respectively $2.48(0.88)$, $1.98(0.31)$, $1.98(0.23)$, $1.98(0.31)$ for MW method, and $2.56(0.88)$,  $3.2(0.75)$, $3.46(0.97)$, $3.7(0.97)$ for BS method. 
    Here, if a discovered range has size no larger than twice the smallest window size and it contains a true change point, it is regarded as a successfully detected change point.  
 
The simulation results shows that MW is more robust and computationally efficient than BS method. As was pointed out in the previous section, MW is robust because it looks into the data at different resolutions, thus reducing the risks of overfitting or underfitting (which the BS method suffers from).

\begin{figure}
\centering
  \includegraphics[width=0.8\linewidth]{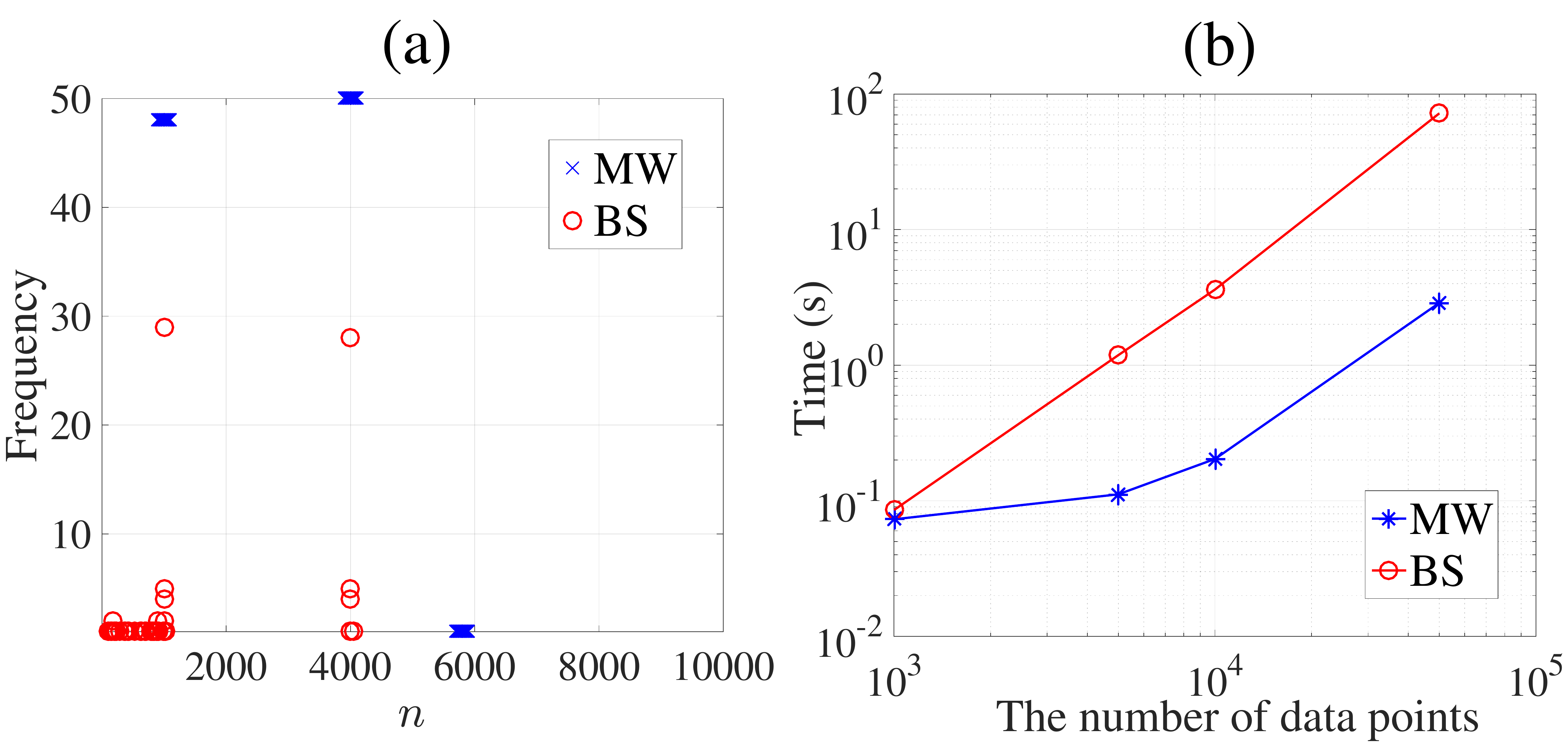}
  \vspace{-0.1in}
  \caption{(a) Frequencies of detected change points (or its ranges) by BS and MW methods,  and (b) log-log plot of the computation time on multiple change points analysis }
  \label{fig:compare}
  \vspace{-0.1in}
\end{figure}

\subsection{Eastern US temperature from 1895 to 2015}

In this subsection, we investigate the temporal variability of the summer-time temperature over the Eastern US  for 1895-2015 (plotted in Fig.~\ref{fig:AMO:yearParcorr}(a)) with our change detection algorithm. The temperature data is obtained from National Climatic Data Center (NCDC, http://www.ncdc.noaa.gov/) and averaged over the Eastern US (east of $100^{\circ}$W). Fig.~\ref{fig:AMO:yearParcorr} shows the data  and its sample partial autocorrelations, from which we recognize the data as independent.
We choose $\m = 7$, and try a range of penalty terms $f(N) = j \log \log N, j=1,...,5$. We start with $j=1,2$; the penalty is so small that it gives the maximally possible $7$ change points. Then we increase $f(N)$ to $3 \log \log N$, and obtain $5$ change points at  years 1901, 1929, 1944, 2009, 2012 (marked in solid lines in Fig.~\ref{fig:AMO:yearData}(a)). If $f(N)$ is increased to $4 \log \log N$, the change points are the years 1929, 1944, 2004 (marked in dashed lines in Fig.~\ref{fig:AMO:yearData}(a)). If $f (N )$ is further increased to $j \log \log N, j\geq 5$, there is no change point detected. 
The segmentation of the time series of the Eastern US temperature over the past century matches the phase shift of the Atlantic Multi-decadal Oscillation (AMO), defined as the North Atlantic sea surface temperature after removing the long-term warming trend \cite{sutton2005atlantic}. As seen from Fig.~\ref{fig:AMO:yearData}(b), since the early $20$th century, there are warm phases from 1929 to 1960 and from 1990 to 2015, and cool phases from 1901 to 1929 and from 1965 to 1990, in synchrony with the segmentation of the Eastern US temperature time series defined by the change points. As the ocean has much larger heat capacity than the continent, this implies that the multi-decadal variability of Eastern US temperature is modulated by the AMO. The dynamic link between AMO and Eastern US climate has previously been reported. For example, based upon a global climate model, it was  indicated in \cite{sutton2005atlantic,sutton2007climate} that the AMO plays an important role in driving the summer-time temperature in the Eastern US. This validates our conclusion derived from the change point detection algorithm.

\begin{figure}
\centering
  \includegraphics[width=0.6\linewidth]{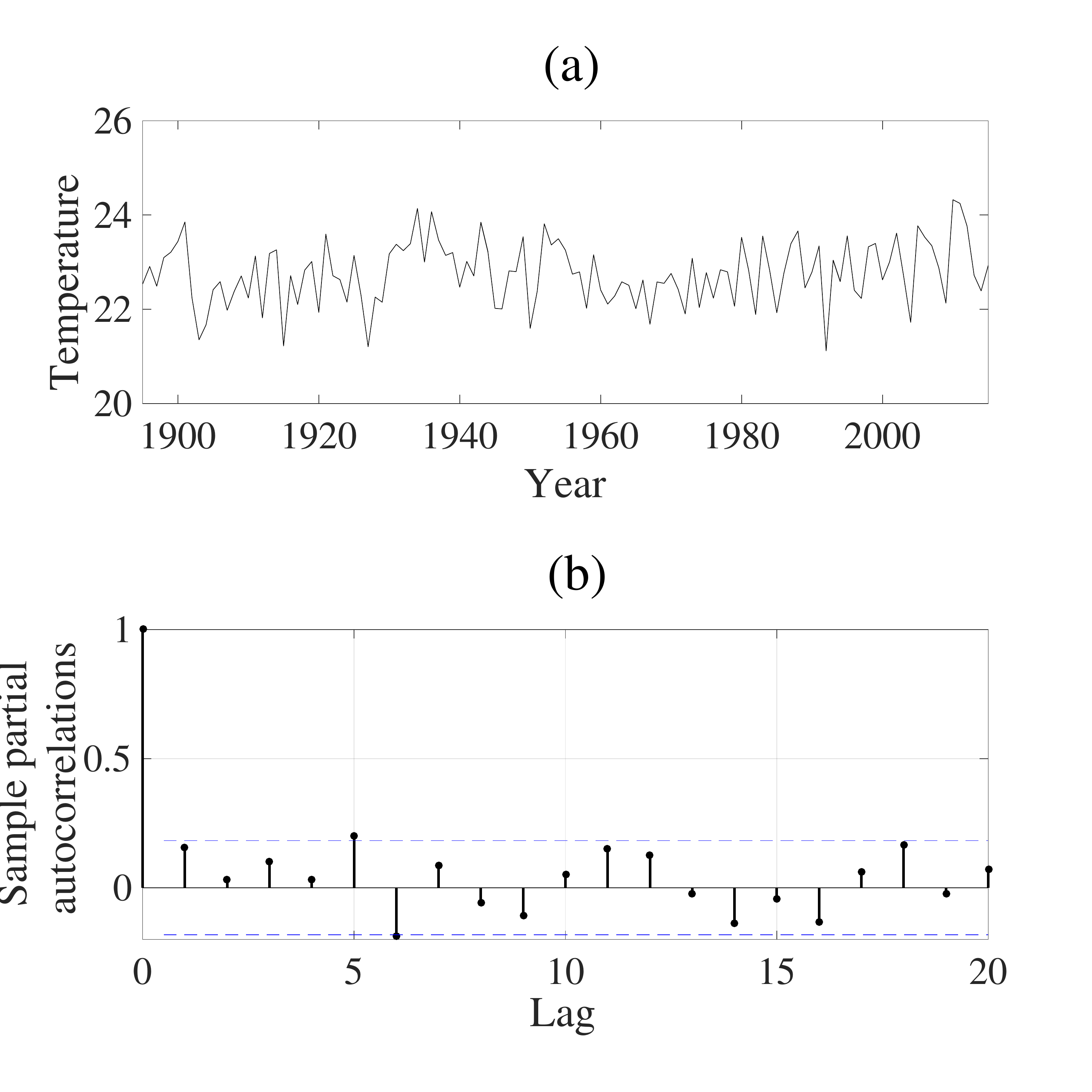}
  \vspace{-0.2in}
  \caption{(a) 1895-2015 summer-time temperature over the Eastern US (unit: $^{\circ}$C), and (b) its sample partial autocorrelations}
  \label{fig:AMO:yearParcorr}
  \vspace{-0.2in}
\end{figure}

\begin{figure}
\centering
  \includegraphics[width=0.6\linewidth]{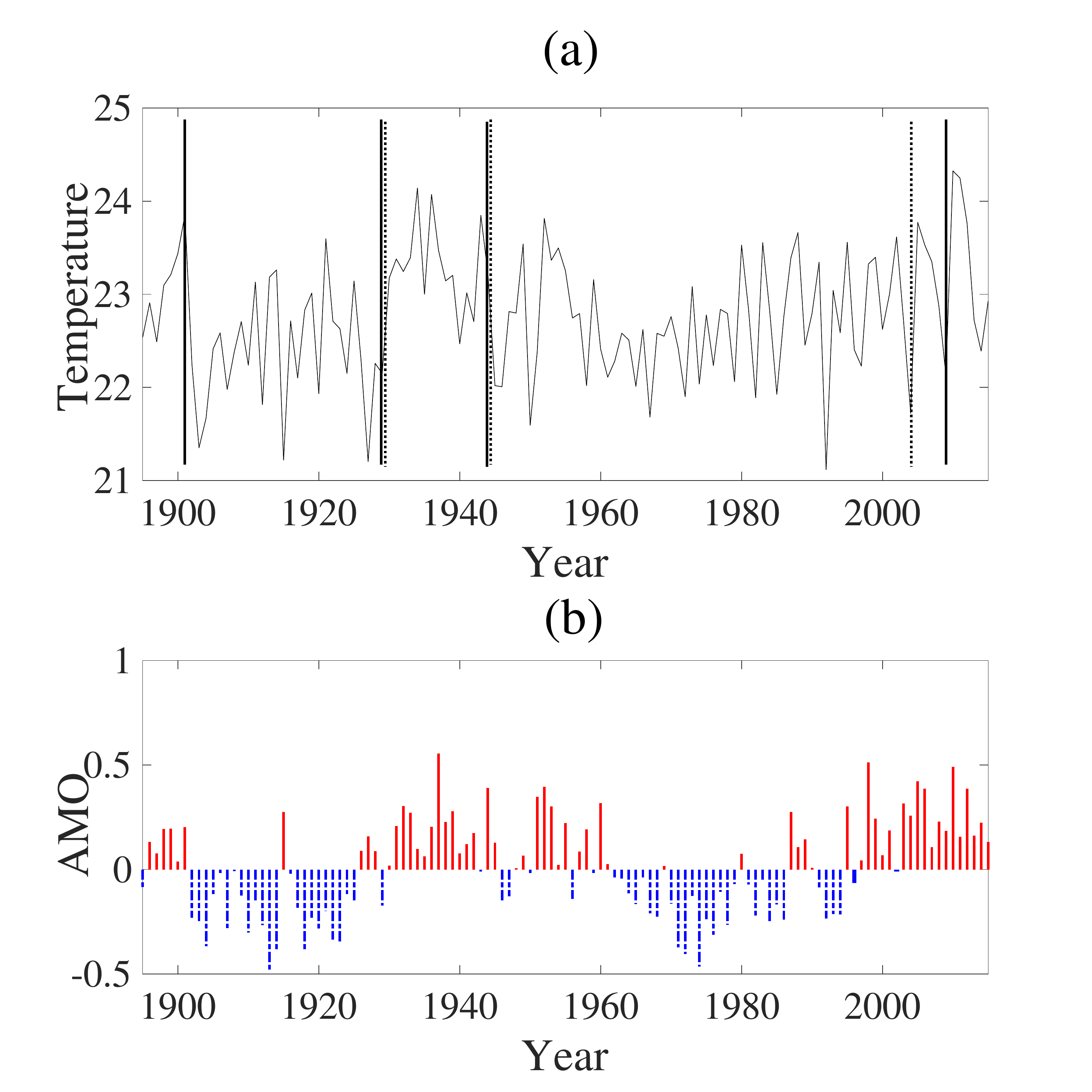}
  \vspace{-0.2in}
  \caption{(a) Detected change points of the Eastern US temperature, and (b) phase shifts of the AMO }
  \label{fig:AMO:yearData}
  \vspace{-0.2in}
\end{figure}

\subsection{El Nino data from 1854 to 2015}
As the largest climate pattern,  El Nino serves as the most dominant factor of oceanic influence on climate. The NINO3 index, defined as the area averaged sea surface temperature from $5^{\circ}$S-$5^{\circ}$N and $150^{\circ}$W-$90^{\circ}$W, is calculated from HadISST1 from 1854 to 2015 \cite{rayner2003global}, as shown in Fig.~\ref{fig:NINO:monthData}(a) (with 1944 points). By looking at the partial autocorrelation of the complete dataset in Fig.\ref{fig:NINO:monthData}(b), we tentatively set autoregression order $L = 2$ (in fact, we also experimented the cases $L = 3,4,5$ and the final results did not differ much). We apply Algorithm~\ref{algo:multiWindow} with window sizes $300, 250, 200, 150, 100, 50$, and  $\m = \lfloor N/300 − 1 \rfloor = 5$ (where$\lfloor a \rfloor$ denotes the largest integer that is no larger than $a$). 
We start with $f (N ) = 2 \log \log N$ and obtain the score plot as shown in Fig.~\ref{fig:NINO:score}(a). The plots show that the time period from June 1979 to September 1987 most likely contains one change point. We change the penalty to smaller or larger values, or use other window sizes, and found that the range is detected most of the time. 
In fact, we can trace how the AR coefficients change in Fig.~\ref{fig:NINO:score}(b), where each point is the AR coefficient estimated from a sliding window of size 300 and sliding step size 20. In other words, the windows are $\{X_1,\ldots,X_{300}\},\{X_{21},\ldots,X_{320}\},\ldots,\{X_{1641},\ldots,X_{1941}\}$. The green diamond, blue star, and red circle indicate respectively the first 37 windows, the second 37 windows, and the last 9 windows. As illustrated from the plot, the red circles deviate nontrivially from other points, which means that the data has a structural change after 74 windows, and that time is exactly the year 1979. The shift of the Pacific Decadal Oscillation (PDO) from a long cold phase (1940-1978) to a warm phase (1979-present) is likely to explain why this year is unique in the past 150 years. The PDO can have a strong influence on the climate in the Northern hemisphere, including the drought frequency in the North America \cite{mccabe2004pacific}, ecosystem productivity \cite{francis1998effects}, as well as the Bermuda High pressure system in Atlantic ocean \cite{li2012variation} . 

\begin{figure}
\centering
  \includegraphics[width=0.6\linewidth]{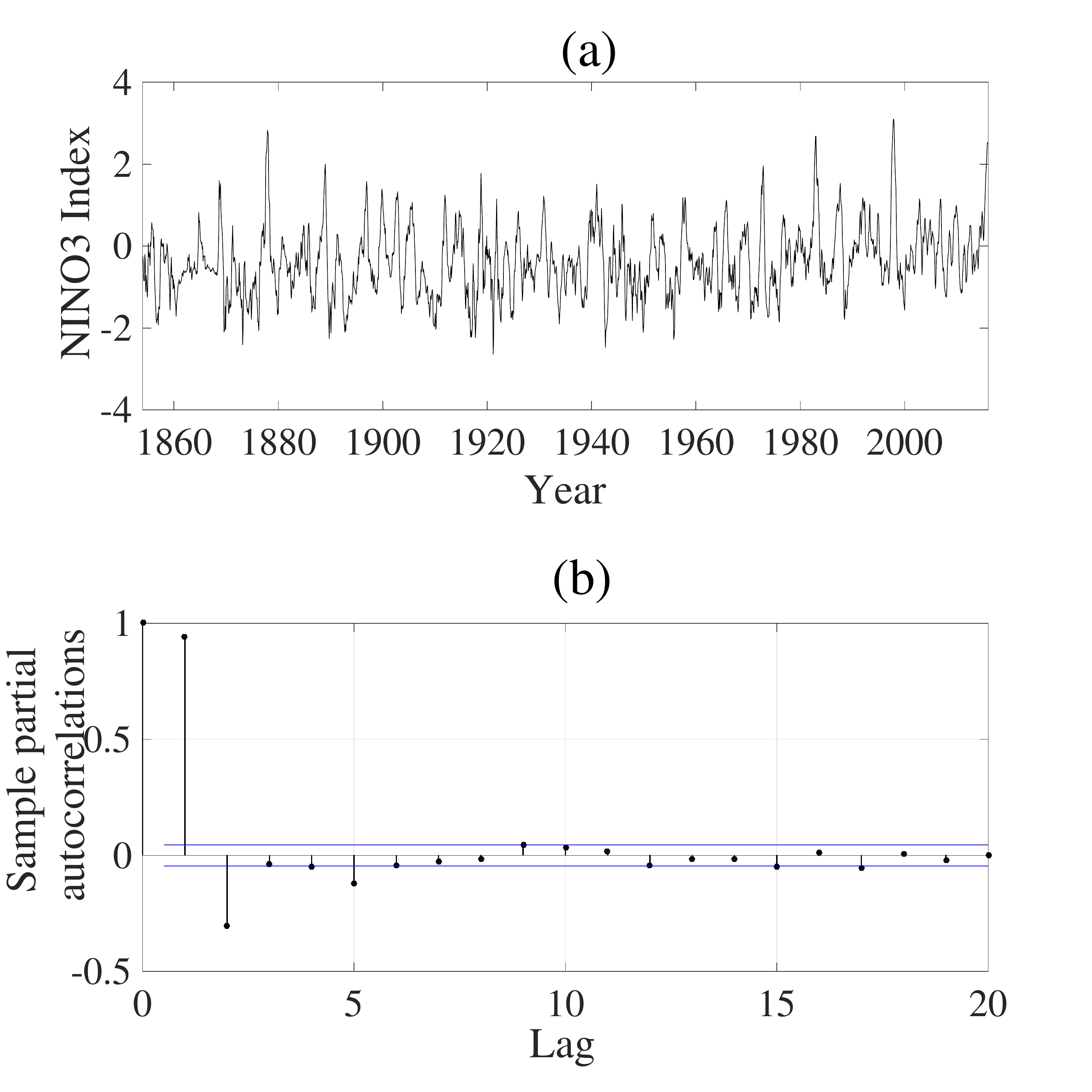}
  \vspace{-0.2in}
  \caption{(a) Monthly El Nino (Nino3) index  from 1854 to 2015, and (b) its sample partial autocorrelations}
  \vspace{-0.2in}  
  \label{fig:NINO:monthData}
\end{figure}

\begin{figure}
\centering
  \includegraphics[width=0.8\linewidth]{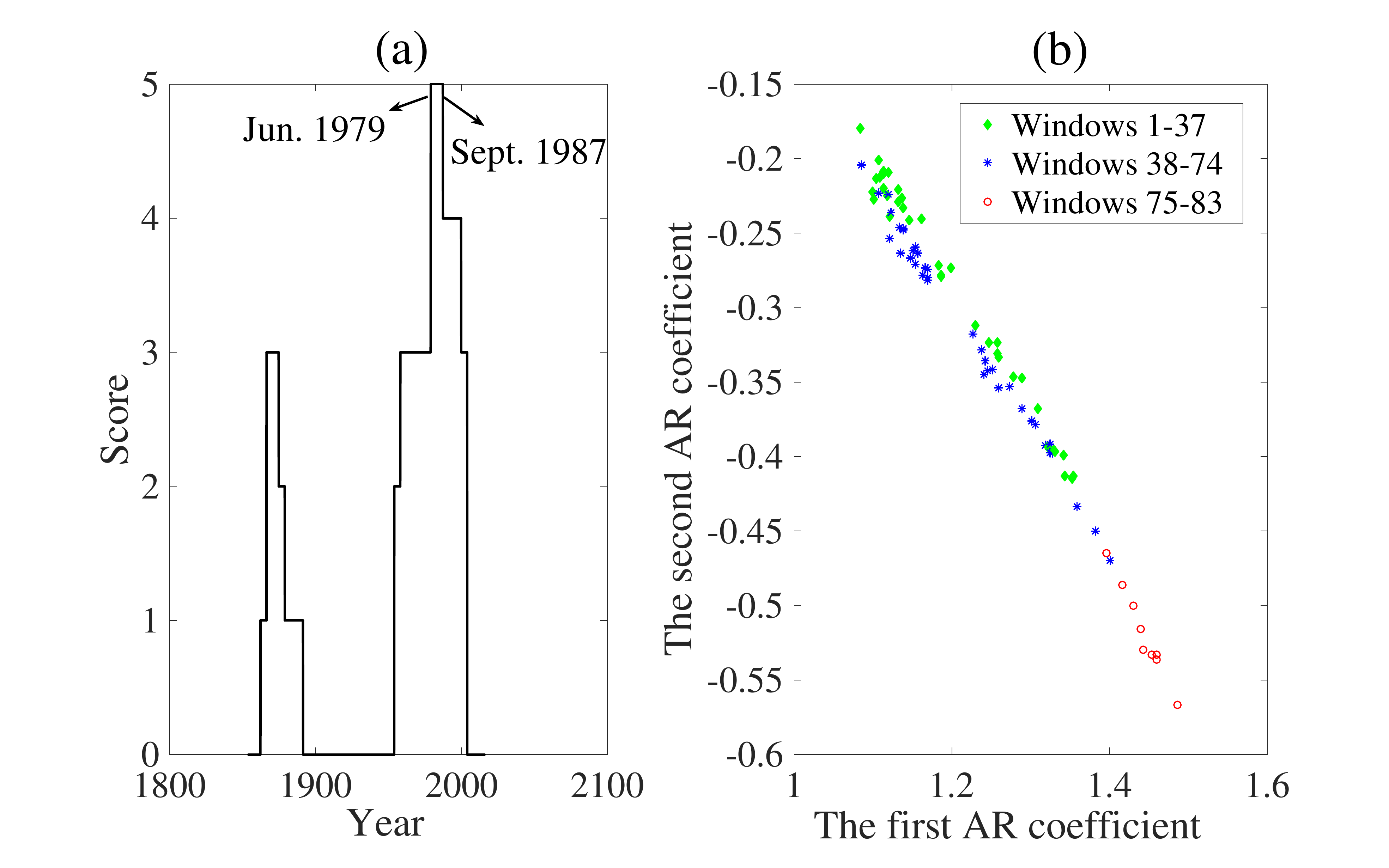}
  \vspace{-0.15in}
  \caption{(a) Score plot of El Nino data obtained from Algorithm~\ref{algo:multiWindow} which indicates the ranges of change points, and (b) the trace plot that illustrates how the coefficients of AR(2) vary with time  }
  \label{fig:NINO:score}
  \vspace{-0.2in}
\end{figure}

\section{Conclusion}

Although prior work has extensively focused on the consistency issues involved in multiple change points analysis, little attention has been paid to deriving strong consistency. This work investigated the necessary and  sufficient conditions under which a model selection criterion is strongly consistent. 
Our analysis is under the assumption of the independence of data, and for the quadratic loss function. Nevertheless, it appears that our proposed technical tools can be applied to studying richer data structures. 
Furthermore, we modeled a general stochastic process by segment-wise autoregressions, and proposed an effective and efficient multi-window technique for change detection. The main idea was to turn the change detection problem into that of independent data scenario. 
Different window sizes are applied and properly combined to achieve a good trade-off between estimation accuracy and the resolution of detection.
Generalization to other loss functions or procedures is possible 
and will be considered a future work. 

\pagebreak

\appendices
\section*{Appendix }


Let $S^{(k)}_{n_1:n_2} = \sum_{n=L_{k-1}+n_1+1}^{L_{k-1}+n_2} (X_n-\mu_k)$,
and $S^{(k)} = S^{(k)}_{0:N_k}$.
Let $S^{(k_1,k_2)}_{n_1:n_2}=S^{(k_1)}_{n_1:N_{k_1}}+S^{(k_1+1)}+\cdots+S^{(k_2-1)}+S^{(k_2)}_{0:n_2}$ for $k_1<k_2$
and $S^{(k_1,k_2)}_{n_1:n_2}=S^{(k_1)}_{n_1:n_2}$ for $k_1=k_2$.

We define the within-segment loss 
$Q^{(k)}_{n_1:n_2} =$ $\textit{Loss}_q(x_{L_{k-1}+n_1+1},\ldots,x_{L_{k-1}+n_2})$, and 
the cross-segment loss 
$Q^{(k_1,k_2)}_{n_1:n_2} = \textit{Loss}_q(x_{L_{k_1-1}+n_1+1},\ldots,x_{L_{k_2-1}+n_2})$. 
Let 
\begin{align}
	g^{(k_1,k_2,k_3)}_{n_1,n_2,n_3}=Q^{(k_1,k_3)}_{n_1:n_3}-(Q^{(k_1,k_2)}_{n_1:n_2}+Q^{(k_2,k_3)}_{n_2:n_3}), \label{eq_break}
\end{align}
referred to as the \textit{decomposition gain}. 
In the case of $k_1=k_2<k_3,n_2=N_{k_1}$, 
we denote $g^{(k_1,k_2,k_3)}_{n_1,n_2,n_3}$ by $g^{(k_1,k_3)}_{n_1,n_3}$;
In the case of $k_1=k_2=k_3$, we denote $g^{(k_1,k_2,k_3)}_{n_1,n_2,n_3}$ by $g^{(k_1)}_{n_1,n_2,n_3}$. Thus,
\begin{align}
		g^{(k)}_{n_1,n_2,n_3}
	=&Q^{(k)}_{n_1:n_3}-(Q^{(k)}_{n_1:n_2}+Q^{(k)}_{n_2:n_3}) ,\nonumber \\
		g^{(k,k+1)}_{n_1,n_2}
	=& Q^{(k,k+1)}_{n_1:n_2} - 
	(Q^{(k)}_{n_1:N_{k}}+Q^{(k+1)}_{0:n_2})	. \label{eq_break}
\end{align} 

If $n_1 \geq n_2$ or $n_2 \geq n_3$ in the above definitions, the corresponding values are understood to be zeros. 
For each $d=1, \ldots, D$, let $X_{n,d}$ and $S^{(k)}_{n_1:n_2,d}$ denote the $d$th component of $X_n$ and $S^{(k)}_{n_1:n_2}$, respectively.

\vspace{0.3cm}
{\bf \textit{Technical Lemmas}}
\vspace{0.3cm}

\begin{lemma} \label{lemma_5}
	Suppose that $n_2 \geq n_1 \geq 34$. Then 
	$$\frac{n_1}{n_1+n_2} \log\log(n_2) \leq \frac{1}{2}\log \log (n_1).$$
\end{lemma}
\begin{proof}
	Define $h(x)=x \log\log(n-x) / \{n \log\log (x)\}$ for $3 \leq x \leq n/2$. Let $y=\log(n-x) \geq \log(n/2)$. By simple calculation, $dh(x)/dx \geq 0$ is equivalent to $y \log(y) \geq (2*[1-\{\log(x)\}^{-2}])^{-1}$ which can be guaranteed if $y \geq \log(17)$. 
	Thus, for $n \geq 34$, $h(x)$ is an increasing function on $x \in [3, n/2]$ with maximum value $1/2$. Lemma~\ref{lemma_5} follows from taking $n=n_1+n_2,x=n_1$. 
\end{proof}


\vspace{0.2cm}

\begin{lemma} \label{lemma_2}
	Suppose that $\{X_n:n=1,\ldots,N_1\}$ and $\{X_n:n=N_1+1,\ldots,N\}$ are independent random variables from the same distribution $\mathcal{G}$, with mean $\mu$ and variance $V$. Let $N_2 = N-N_1$. 
	Assume that $N_1,\, N_2 \rightarrow \infty$ as $N \rightarrow \infty$, and $N_1,N_2$ depend only on $N$. 
	Then $g^{(1,2)}_{0,N_2}$ (the decomposition gain)
	converges in distribution to $Z^\T V Z$, where $Z \in \mathcal{N}_D(0,I)$. 
	Moreover,
	\begin{equation} 
	 \limsup_{N\rightarrow \infty} 
	 \frac{g^{(1,2)}_{0,N_2}}{\log\log (\min\{N_1,N_2\})} =C \quad (a.s.) \label{eq2}
	\end{equation}
	for some positive constant $C \leq 8 \ tr(V)$.
\end{lemma}

\begin{proof}
	Denote $P=\{1,2,\ldots,N\}, P_1=\{1,2,\ldots,N_1\}, P_2=\{N_1+1,2,\ldots,N\}$.
	By direct calculation, we obtain 
	\begin{align}
		g^{(1,2)}_{0,N_2}  
		&=\biggl(\sum_{n \in P} |X_n|^2-N\biggl|\frac{\sum_{n \in P}X_n}{N} \biggr|^2\biggr) 
		- \biggl(\sum_{n\in P_1} |X_n|^2-N_1\biggl|\frac{\sum_{n\in P_1} X_n}{N_1}\biggr|^2 
		+ \sum_{n\in P_2} |X_n|^2-N_1\biggl|\frac{\sum_{n\in P_2}X_n}{N_2}\biggr|^2 \biggr) \nonumber \\
		&=\frac{1}{N_1 N_2 N} \left| N_2 \sum_{n\in P_1} X_n - N_1 \sum_{n^{'} \in P_2}X_{n^{'}}\right|^2 
		= \left|\sqrt{\frac{N_2}{N}} \frac{\sum_{n\in P_1} X_n}{\sqrt{N_1}}-\sqrt{\frac{N_1}{N}} \frac{\sum_{n\in P_2} X_n}{\sqrt{N_2}}\right|^2  \label{eq55} \\
		&= \left|\sqrt{\frac{N_2}{N}} \frac{S^{(1)}_{0:N_1}+N_1\mu}{\sqrt{N_1}}-\sqrt{\frac{N_1}{N}} \frac{S^{(2)}_{0:N_2}+N_2\mu}{\sqrt{N_2}}\right|^2 
		= \left|\sqrt{\frac{N_2}{N}} Y^{(1)}_N- \sqrt{\frac{N_1}{N}} Y^{(2)}_N \right|^2 \label{eq1}
	\end{align}
	where $Y^{(k)}_N = S^{(k)}_{0:N_k}/\sqrt{N_k}, k=1,2$. 
	By the central limit theorem (CLT), $Y^{(1)}_N, Y^{(2)}_N$ converge in distribution to two independent $\mathcal{N}_D(0,V)$ random variables, respectively denoted by $Y^{(1)}, Y^{(2)}$.
	Therefore, 
	$$ \sqrt{\frac{N_2}{N}} Y^{(1)}_N - \sqrt{\frac{N_1}{N}} Y^{(2)}_N  = \sqrt{\frac{N_2}{N}} Y^{(1)} - \sqrt{\frac{N_1}{N}} Y^{(2)} + o_p(1)$$
	converges in distribution to a random variable $W \sim \mathcal{N}_D(0,V)$. Let  $W = V^{1/2}Z$, then $Z \sim  \mathcal{N}_D(0,I)$.
	It follows that  $g^{(1,2)}_{0,N_2}$ in (\ref{eq1}) converges to $Z^\T V Z$.
	Furthermore, by the law of the iterated logarithm, 
	\begin{align} \label{eq3}
		\limsup_{N_k\rightarrow \infty}Y^{(k)}_{N,d}/\sqrt{2 V_{dd} \log\log N_k} = 1 , \quad (a.s.), \quad k=1,2 , \, d=1,\ldots,D.
	\end{align}
	where $Y^{(k)}_{N,d}$ and $V_{dd}$ denote the $d$th entry of $Y^{(k)}_{N}$ and the $(d,d)$th entry of $V$, respectively. 
	Note that 
	$$ \biggl| \sqrt{\frac{N_2}{N}} Y^{(1)}_{N,1} - \sqrt{\frac{N_1}{N}} Y^{(2)}_{N,1} \biggr|^2 
	\leq g^{(1,2)}_{0,N_2} 
	\leq \sum_{d=1}^D \biggl\{ \sqrt{\frac{N_2}{N}}|Y^{(1)}_{N,d}| + \sqrt{\frac{N_1}{N}} |Y^{(2)}_{N,d}| \biggr\}^2 $$
	where the second inequality follows from triangle inequality.
	We infer from (\ref{eq3}) that for any fixed $\delta \in (0,1)$, almost surely
	\begin{align} 
	 &g^{(1,2)}_{0,N_2} \geq \biggl(\sqrt{\frac{2N_2}{N}(1-\delta) V_{11} \log\log N_1} + \sqrt{\frac{2N_1}{N} (1-\delta) V_{11} \log\log N_2}\biggr)^2  \quad i.o. \label{sandwich1} \\
	 &\limsup_{N \rightarrow \infty} g^{(1,2)}_{0,N_2} \biggl[
	   \sum_{d=1}^D \biggl(\sqrt{\frac{2N_2}{N} V_{dd} \log\log N_1} + \sqrt{\frac{2N_1}{N} V_{dd} \log\log N_2}\biggr)^{2} \biggr]^{-1} 
	    \leq 1	  \label{sandwich2}
	\end{align}
	From (\ref{sandwich1}), it is easy to observe (with $\delta = 1/2$) that
	\begin{align} 
	 &g^{(1,2)}_{0,N_2} > V_{11} \log\log (\min\{N_1,N_2\})  \quad i.o. \label{sandwich3} 
	\end{align}
	It follows from (\ref{sandwich2}) and Lemma~\ref{lemma_5} that
	\begin{equation} 
		\limsup_{N \rightarrow \infty} g^{(1,2)}_{0,N_2} \ \biggl[
	   8 \ tr(V) \log\log (\min\{N_1,N_2\}) \biggr]^{-1} 
	    \leq 1	\quad (a.s.)   \label{sandwich4}
	\end{equation}
	Furthermore, since $V_{dd} > 0$, Inequalities (\ref{sandwich3}) and (\ref{sandwich4}) imply the desired equality (\ref{eq2}). 
	
\end{proof}

\begin{remark}
Lemma~\ref{lemma_2} implies that splitting a sequence of i.i.d. points into two halves increases the goodness of fit (measured by quadratic loss) by $O_p(1)$. Therefore,  an AIC-like criterion (with constant penalty) always produces a non-vanishing overfitting probability.	
\end{remark}

\vspace{0.2cm}

\begin{lemma} \label{lemma_4}
	 Under model assumption (M.2), 
	 for any $j \in \{1,\ldots,M_0\}$ and $n_1,n_2$ satisfying 
	 $$c^{-1} \leq \frac{n_1}{ N_j} \leq 1 , \quad c^{-1} \leq \frac{n_2}{N_{j+1} }\leq 1, $$ 
	 where $c>1$ is some constant, 
	 we have 
	\begin{align} \label{eq90}
		g^{(j,j+1)}_{N_j-n_1,n_2} > \frac{1}{3} |\mu_{j}-\mu_{j+1}|^2 \min\{n_1,n_2\} \quad \textrm{ for sufficiently large $N$  } \, (a.s.)
	\end{align} 
	In other words, for each $\omega$ from a set of probability one, there exists a positive constant $N_{\omega}$ such that Inequality (\ref{eq90}) holds for all $N > N_{\omega}$.
\end{lemma}
\begin{proof}
	From Equation (\ref{eq55}) (note that its derivation does not require the two segments to have the same mean), we obtain 
	$$g^{(j,j+1)}_{N_j-n_1,n_2} = \biggl|\sqrt{\frac{n_2}{n}} Y^{(1)}_n- \sqrt{\frac{n_1}{n}} Y^{(2)}_n + \sqrt{\frac{n_1 n_2}{n}} (\mu_j-\mu_{j+1}) \biggr|^2,
	$$ 
	where 
	$$n=n_1+n_2, \quad
	Y^{(1)} = \sum_{i=L_{j}-n_1+1}^{L_{j}}\frac{X_i-\mu_j}{\sqrt{n_1}}, \quad
	Y^{(2)} = \sum_{i=L_{j}+1}^{L_{j}+n_2}\frac{X_i-\mu_{j+1}}{\sqrt{n_2}}.
	$$   
	By triangle inequality $g^{(j,j+1)}_{N_j-n_1,n_2} \geq (|B|-|A|)^2 $, where 
	$$
	A = \sqrt{\frac{n_2}{n}} Y^{(1)}- \sqrt{\frac{n_1}{n}} Y^{(2)}, \quad
	B= \sqrt{\frac{n_1 n_2 }{n}} (\mu_j-\mu_{j+1}).
	$$ 
	By Strassen's invariance principle \cite[Chapter 5]{billingsley2013convergence}, for each individual $\omega$ in a set of probability one, for each $d=1,\ldots,D$   
	\begin{align*} 
		&\limsup\limits_{N_j\rightarrow \infty} \frac{\sum_{i=L_{j}-n_1(\omega)+1}^{L_{j}}(X_{i,d}(\omega)-\mu_{j,d})}{\sqrt{2 V_{j,dd} N_j \log\log N_j}} \leq 1, \quad 
		\limsup\limits_{N_{j+1}\rightarrow \infty} \frac{\sum_{i=L_{j}+1}^{L_{j}+n_2(\omega)}(X_{i,d}(\omega)-\mu_{j+1,d})}{\sqrt{2 V_{j+1,dd} N_{j+1} \log\log N_{j+1}}} \leq 1	
	\end{align*}
	which implies that 
	\begin{align}
		\frac{Y^{(1)}_{d}(\omega) }{\sqrt{2 V_{j,dd} \log\log n_1} }
		=\frac{\sqrt{n_1} \ Y^{(1)}_{d}(\omega) }{\sqrt{2 V_{j,dd} N_j \log\log N_j}} 
		\frac{\sqrt{N_j \log\log N_j}}{\sqrt{n_1 \log\log n_1 }} 
		\leq  \sqrt{c+1},
		\quad  \frac{Y^{(2)}_{d}(\omega) }{\sqrt{2 V_{j+1,dd} \log\log n_2 } } \leq  \sqrt{c+1} \nonumber 
	\end{align}
	for sufficiently large $N$ (thus $N_j,N_{j+1}$). For brevity, we have simplified $n_1(\omega),n_2(\omega)$ to $n_1,n_2$.  
	From the above inequalities and Lemma~\ref{lemma_5}, we obtain  
	\begin{align}
		|A|^2 &\leq \sum_{d=1}^D \left(\sqrt{\frac{n_2}{n}} \sqrt{2(c+1) V_{j,dd}\log\log n_1} + \sqrt{\frac{n_1}{n}} \sqrt{2(c+1) V_{j+1,dd} \log\log n_2} \right)^2 \nonumber \\
		&\leq \sum_{d=1}^D 2 (c+1) \max\{V_{j,dd},V_{j+1,dd}\} \left(\sqrt{\frac{n_2}{n}\log\log n_1}  + \sqrt{\frac{n_1}{n}\log\log n_2}  \right)^2 \nonumber \\
		&< \sum_{d=1}^D 8(c+1) \max\{V_{j,dd},V_{j+1,dd}\} \log\log (\min\{n_1,n_2\})  \label{eq_8}
	\end{align}
	for sufficiently large $N$ almost surely. 
	It follows from 
	$$|B| = \sqrt{\frac{n_1 n_2}{n}} |\mu_j-\mu_{j+1}| \geq \sqrt{\frac{\min\{n_1,n_2\}}{2}} |\mu_{j}-\mu_{j+1}|
	$$ 
	that $g^{(j,j+1)}_{N_j-n_1,n_2} > |\mu_{j}-\mu_{j+1}|^2 \min\{n_1,n_2\}/3$ for sufficiently large $N$ almost surely. 
\end{proof}

\vspace{0.4cm}
{\bf \textit{Proof of Theorem~\ref{thm:independency}}}
\vspace{0.2cm}

%
	For the case $L=0$, $\{Y_n:n=1,\ldots,N\}$ are independent, and $\hat{\psi}_1 = \sum_{n=1}^{N_1} Y_n/N_1$, $\hat{\psi}_2 = \sum_{n=N-N_2+1}^{N} Y_n/N_2$. Thus, $\sqrt{N_1}(\hat{\psi}_1-\psi)$ and $\sqrt{N_2}(\hat{\psi}_2-\psi)$ converge to Gaussian random variables that are independent. It remains to prove for the case $L>0$. 
	Choose $N_{1}^{'}, N_{2}^{'}$ such that $N_{1}^{'}/N_1, N_{2}^{'}/N_2 \rightarrow 1, N_1-N_{1}^{'}, N_2-N_{2}^{'} \rightarrow \infty$.
	Let $\hat{\psi}_1^{'}, \hat{\psi}_2^{'}\in \mathcal{R}^{L+1}$ respectively denote the estimated filters from $\{X_1,\cdots,X_{N_1^{'}}\}$ and $\{X_{N-N_2^{'}+1},\cdots,X_{N}\}$ using the least squares method. 
	It is well known that $\sqrt{N_1^{'}}(\hat{\psi}_1^{'}-\psi), \sqrt{N_2^{'}}(\hat{\psi}_2^{'}-\psi)$ respectively converge in distribution to $Z_1,Z_2 \sim \mathcal{N}(0,\sigma^2(\Gamma_L^{*})^{-1})$, where 
	$
	\Gamma_L^{*} = 
	\begin{pmatrix}
	1  & 0 \\
	0 &   \Gamma_L    
	\end{pmatrix}
	$ 
	and $\Gamma_L$ is the covariance matrix of order $L$ \cite[Appendix 7.5]{box2011time}. 
	Because ${X_n}$ is strongly mixing under Assumption (A.1) 
	 \cite{athreya1986note},  $\sqrt{N_1^{'}}(\hat{\psi}_1^{'}-\psi)$ and $\sqrt{N_2^{'}}(\hat{\psi}_2^{'}-\psi)$ are asymptotically independent. Thus, $Z_1$ and $Z_2$ are independent. 
	It remains to prove that $\sqrt{N_1}(\hat{\psi}_1-\psi)=\sqrt{N_1^{'}}(\hat{\psi}_1^{'}-\psi)+o_p(1)$ and $\sqrt{N_2}(\hat{\psi}_2-\psi)=\sqrt{N_2^{'}}(\hat{\psi}_2^{'}-\psi)+o_p(1)$.
	We prove the former equation since the latter one can be similarly proved.
	Let  
\begin{align} \label{det}
Z_1=
\begin{pmatrix}
	y_{N_1-1}  & \cdots & y_{N_1-L}\\
	\vdots &   \ddots  & \vdots 	\\
	y_{L}  &  \cdots & y_{1}
\end{pmatrix},
W_1=
\begin{pmatrix}
	y_{N_1}  \\
	\vdots  \\
	y_{L+1}
\end{pmatrix},
E_1=
\begin{pmatrix}
	\v_{N_1}  \\
	\vdots  \\
	\v_{L+1}
\end{pmatrix}.
\end{align}
Since $\hat{\psi}_1$ is estimated from least squares method, it can be written in the matrix form  $\hat{\psi}_1=(Z_1^\T Z_1)^{-1} Z_1^\T W_1 = (Z_1^\T Z_1)^{-1} Z_1^\T (Z_1 \psi_1 + E_1) = \psi_1 + (Z_1^\T Z_1)^{-1} Z_1^\T E_1 $. 
We similarly define $Z_1^{'},W_1^{'}, E_1^{'}$ by substituting $N_1$ with $N_1^{'}$ in (\ref{det}), and write $\hat{\psi}_1^{'}= \psi_1 + \{(Z_1^{'})^\T Z_1^{'}\}^{-1} (Z_1^{'})^\T E_1^{'}$.
	Therefore,
	\begin{align} \label{eq:101}
		\sqrt{N_1} (\hat{\psi}_1 - \psi_1) 
			&= \biggl(\frac{Z_1^\T Z_1}{N_1} \biggr)^{-1} \frac{ Z_1^\T E_1 }{\sqrt{N_1}} , \quad
		\sqrt{N_1} (\hat{\psi}_1^{'} - \psi_1) 
			= \biggl\{\frac{(Z_1^{'})^\T (Z_1^{'})}{N_1} \biggr\}^{-1} \frac{ (Z_1^{'})^\T E_1^{'} }{\sqrt{N_1}} .
	\end{align}
	Recall that 
	$$
	 \frac{Z_1^\T Z_1}{N_1}, \, 
	\frac{(Z_1^{'})^\T Z_1^{'} }{ N_1^{'}}, \,
	\frac{ Z_1^\T Z_1 - (Z_1^{'})^\T Z_1^{'} }{ N_1-N_1^{'} }
	\rightarrow \Gamma_L^* \,\,
	\textrm{in probability, as }  \, N \rightarrow \infty, 
	$$  
    where $\Gamma_L$ is the covariance matrix of order $L$, and 
	\begin{align} \label{eq:103}
	\frac{Z_1^\T E_1 }{ \sqrt{N_1} }, \, 
	\frac{(Z_1^{'})^\T E_1^{'} }{ \sqrt{N_1^{'}}}, \,
	\frac{Z_1^\T E_1 - (Z_1^{'})^\T E_1^{'} }{ \sqrt{N_1 - N_1^{'}}} \rightarrow \mathcal{N}(0,\sigma^2 \Gamma_L^*)  \,\,
	\textrm{in distribution, as }  \, N \rightarrow \infty, 
	\end{align} 
	due to the central limit theorem for  martingale difference sequences \cite[Appendix 7.5]{box2011time}. Therefore, 
	$$ \frac{ (Z_1^{'})^\T E_1^{'} }{ \sqrt{N_1} } = \frac{ Z_1^\T E_1 + \sqrt{N_1 - N_1^{'}} O_p(1) }{ \sqrt{N_1}}
	= \frac{ Z_1^\T E_1 }{ \sqrt{N_1}} + o_p(1),
	$$ and Equation (\ref{eq:101}) further implies that
	\begin{align} \label{eq:102}
		\sqrt{N_1} (\hat{\psi}_1 - \psi_1) = O_p(1), \quad
		\sqrt{N_1} (\hat{\psi}_1^{'} - \psi_1)	= \biggl\{\frac{(Z_1^{'})^\T (Z_1^{'})}{N_1} \biggr\}^{-1} \frac{ Z_1^\T E_1 }{\sqrt{N_1}} +o_p(1) .	
	\end{align}
	Straightforward calculations using (\ref{eq:101}) and (\ref{eq:102}) give  
	\begin{align*}
		\sqrt{N_1}(\hat{\psi}_1-\psi)-\sqrt{N_1^{'}}(\hat{\psi}_1^{'}-\psi)
		&=\biggl\{\sqrt{N_1}(\hat{\psi}_1-\psi)-\sqrt{N_1}(\hat{\psi}_1^{'}-\psi)\biggr\}
		+\biggl\{\sqrt{N_1}(\hat{\psi}_1^{'}-\psi)-\sqrt{N_1^{'}}(\hat{\psi}_1^{'}-\psi) \biggr\} \\
		&= \biggl(\frac{Z_1^\T Z_1}{N_1} \biggr)^{-1} \frac{ Z_1^\T E_1 }{\sqrt{N_1}} - \biggl\{ \frac{(Z_1^{'})^\T (Z_1^{'})}{N_1^{'}} \biggr\}^{-1} \frac{ Z_1^\T E_1 }{\sqrt{N_1}} + o_p(1) +  \biggl(\sqrt{\frac{N_1}{N_1^{'}} } - 1 \biggr) O_p(1) \\
		&= \biggl[ \biggl(\frac{Z_1^\T Z_1}{N_1}\biggr)^{-1} - \biggl\{\frac{(Z_1^{'})^\T (Z_1^{'})}{N_1}\biggr\}^{-1}  \biggr] O_p(1) + o_p(1).
	\end{align*}
	To finish the proof, it suffices to prove that 
	$$
	\biggl(\frac{Z_1^\T Z_1}{N_1}\biggr)^{-1} - \biggl(\frac{(Z_1^{'})^\T (Z_1^{'})}{N_1}\biggr)^{-1}
	$$
	is $o_p(1)$. 
	In fact, from (\ref{eq:103}), the above matrix equals 
	\begin{align*}
	&\biggl\{\frac{(Z_1^{'})^\T (Z_1^{'})}{N_1}\biggr\}^{-1} 
	\biggl\{\frac{(Z_1^{'})^\T (Z_1^{'})-Z_1^\T Z_1}{N_1}\biggr\}
	\biggl\{\frac{Z_1^\T Z_1}{N_1}\biggr\}^{-1} \\
	&=
	\frac{N_1}{N_1^{'}} \biggl\{ \frac{(Z_1^{'})^\T (Z_1^{'})}{N_1^{'}}\biggr\}^{-1} 
	\frac{N_1^{'} - N_1}{N_1}  \biggl\{ \frac{Z_1^\T Z_1 - (Z_1^{'})^\T (Z_1^{'})}{N_1-N_1^{'}}\biggr\}
	\biggl(\frac{Z_1^\T Z_1}{N_1^{'}}\biggr)^{-1}
	= o_p(1).
	\end{align*} 

\vspace{0.4cm}
{\bf \textit{Proof of Theorem~\ref{thm:loglog0changepoint}}}
\vspace{0.2cm}

	We first prove that $f(N)$ should be at least  $\Theta (\log\log N)$ to ensure strong consistency. 
	The event $\hat{M}=0$ implies the event $Q^{(1)}_{0:N/2}+Q^{(1)}_{N/2,N}+f(N) \geq Q^{(1)}_{0:N} $. In other words, $g^{(1)}_{0,N/2,N} > f(N)$ implies the event $\hat{M}\neq 0$. 	
	By Lemma~\ref{lemma_2}, there exists $C_1>0$ such that $g^{(1)}_{0,N/2,N} \geq C_1 \log\log N \ i.o. $ This implies that if $f(N) < C_1 \log\log N$, then $g^{(1)}_{0,N/2,N} > f(N) \ i.o. $ and thus $\hat{M}\neq M \ i.o.$
	 
	On the other hand, the event $\hat{M}>0$ implies the event that there exist $0<n_1<n_2$ such that $g^{(1)}_{0,n_1,n_2} \geq f(N)$ and that $n_1,n_2-n_1 \geq \beta(N) = \Theta(N)$. By a similar derivation to that of (\ref{eq_8}) in Lemma~\ref{lemma_4}, we can show that for sufficiently large $N$
	\begin{align}
	g^{(1)}_{0,n_1,n_2} &<  8(c+1) tr(V_1) \log\log N  \quad (a.s.)  
	\end{align}
	where $c > 1$ is some constant. 
	Thus, given that $f(N)=C_2 \log\log N$ for large enough $C_2>0$, $g^{(1)}_{0,n_1,n_2} < f(N)$ for sufficiently large $N$ almost surely. 
	This implies that  $\hat{M} \overset{a.s.}\longrightarrow 0$ as $N \rightarrow \infty$.

\vspace{0.4cm}
{\bf \textit{Proof of Theorem~\ref{thm:loglog}}}
\vspace{0.2cm}

	We first prove that there is no under-fitting, i.e. $\hat{M} \geq M_0$. It suffices to prove that for each $\omega$ from a set of probability one, there exists a positive integer  $N_{\omega}$  such that for all $N > N_{\omega} $, $\hat{M} \neq m $ for each $m=1,\ldots,M_0-1$. 
	We prove the result by contradiction. 
	Assume that $\hat{M} = m < M_0$. Then there exists at least one detected segment that consists of points from at least two neighboring segments, say the $(j-1)$th and $j$th, and that the numbers of points from the two segments are at least $N_{j-1}/2$ and $N_{j}/2$, respectively. 
	Without loss of generality, we assume $N_1,\ldots,N_{M_0+1}$ to be even. In other words, the points $\{X_n:n=L_{j-1}-N_{j-1}/2+1, \ldots,  L_{j-1}+N_{j}/2\}$ are contained in the $k$th detected segment for some $k = 1,\ldots,m+1$. 
	Following the notation of Algorithm~\ref{algo:oracle}, 
	let $\hat{e}_m$ denote the minimal within-segment quadratic loss given $m$ segments.  
	We consider another configuration of change points: for the set of change points that give $\hat{e}_m$, keep all other segments except for the $k$th segment unchanged, and split the $k$th segment into four segments the middle two of which are $\{X_n:n=L_{j-1}-N_{j-1}/4+1, \ldots,  L_{j-1}\}$ and $\{X_n:n=L_{j-1}+1, \ldots,  L_{j-1} +N_{j}/4\}$. 
	Then the number of segments will increase from $m$ to $m+3$, and we obtain from Lemma~\ref{lemma_4} that for sufficiently large $N$, the increased within-segment quadratic loss is larger than $C_1 \min \{N_{j-1},N_j\} $ almost surely, where the constant $C_1=\underline{\Delta}_{\mu}^2/12$. Since $\hat{e}_{m+3}$ is the global minimum of the within-segment quadratic loss under $m+3$ change points, we obtain 
	\begin{align} \label{eq:new1} 
		\hat{e}_{m} - \hat{e}_{m+3} > C_1 \min \{N_{j-1},N_j\} \quad (a.s.)
	\end{align}
	
	On the other hand, because $m+3 \leq M_{\text{max}}$ and the condition in step 3 of Algorithm~\ref{algo:oracle} is satisfied (since each new segment is at least  $\min_{k=1,\ldots,M_0+1} N_{k}/4 \geq \beta(N)$ for sufficiently large N), $\hat{e}_{m+3}$ is a valid output of Algorithm~\ref{algo:oracle}. Furthermore, the event $\hat{M}=m$ implies the event 
	$\hat{e}_{m} - \hat{e}_{m+3} \leq  3f(N)$. 
	In addition, $3f(N) < C_1 \min \{N_{j-1},N_j\} $ for sufficiently large $N$ due to Assumption~(A.3). Thus, $\hat{e}_{m} - \hat{e}_{m+3} < C_1 \min \{N_{j-1},N_j\} $ which contradicts the inequality in (\ref{eq:new1}). Therefore, $\hat{M} \neq m$ for sufficiently large $N $ almost surely. 
	By similar reasoning we can prove Inequality (\ref{eq91}).
	
	Second, we prove the over-fitting part by contradiction. 
	Assume that $\hat{M} = m > 2M_0$, by the pigeonhole principle there are two detected segments that are adjacent and that belong to the same true segment. Without loss of generality, suppose that $\{ X_{n}:n=\tau+1,\ldots,\tau+n_1 \}$ and $\{X_{n}:n=\tau+n_1+1,\ldots,\tau+n_1+n_2\}$ are from distribution $\mathcal{G}_k$.
	We consider the configuration that merges the aforementioned two segments into one while keeping other segments unchanged. 
	Since $n_1, n_2 \geq \beta(N)=\Theta(N)$, via a similar derivation of (\ref{eq_8}), it can be proved that for sufficiently large $N$
	\begin{align} \label{eq:dj1}
	\hat{e}_{m-1} - \hat{e}_{m} 
	&<  
	C_0 \log \log N \quad (a.s.)  
	\end{align}
	for some constant $C_0 > 1$. 
	On the other hand, the event $\hat{M}=m$ implies that 
	$\hat{e}_{m-1} - \hat{e}_{m} \geq f(N) $. Whenever $f(N) \geq C_0\log\log (N) $, $\hat{e}_{m-1} - \hat{e}_{m} \geq  C_0\log\log (N) $ which contradicts the inequality in (\ref{eq:dj1}). 
	Therefore, we obtain 
	\begin{align*}
		&\P\biggl\{ \limsup\limits_{N \rightarrow \infty} (\hat{M} > 2M_0) \biggl\} 
		\leq \P \biggl\{
		\limsup\limits_{N \rightarrow \infty} 
		(\hat{e}_{m-1} - \hat{e}_{m} < C_0\log\log N)
		\biggr\} = 0 .
	\end{align*} 

\vspace{0.4cm}
{\bf \textit{Proof of Theorem~\ref{thm:strongConsistency}}}
\vspace{0.2cm}

To prove Theorem~\ref{thm:strongConsistency}, we need the following additional technical lemmas. The lemmas serve to enumerate various configurations of change points (events) that will not eventually happen given sufficiently large sample size. Loosely speaking, in those configurations, either ``there exists a detected change point that is redundant'' or ``a true change point is too far away from all the detected change points''. The functionality  of each lemma  will be clearly seen in the final proof of Theorem~\ref{thm:strongConsistency}.  
For notational convenience, for each $k=1,\ldots,M_0+1$, we define $P_{k} = \{L_{k-1}+1,\ldots,L_{k}\}$, and use $X_n^{(k)} \ (n=1,\ldots,N_k)$ to represent the points in the $k$th true segment, namely $\{ X_{L_{k-1}+1}, \ldots,X_{L_{k-1}+N_k} \}$.

\vspace{0.1cm}

\begin{lemma} \label{lemma:simpleSplit}
	For each $k=1,2,\ldots,M_0+1$, let $E_{k,N}$ denote the event that Algorithm~\ref{algo:oracle} produces two neighboring segments that are both subsets of $\{X_n, n \in P_k\}$, the true $k$th segment.
	In other words, 
	\begin{align*}
	E_{k,N} = &\biggl\{ \textrm{There exist integers} \ n_1,n_2,n_3  \textrm{ such that } 0 \leq n_1 < n_2 < n_3 \leq N_k , \textrm{ and }  \\
	&\{X^{(k)}_{n}: n = n_1+1,n_1+2, \ldots,n_2\}, \{X^{(k)}_n: n=n_2+1,n_2+2,\ldots,n_3\} \textrm{ are two detected segments.} \biggr\}
	\end{align*}
	Assume that 
	\begin{align}
		f(N) \geq  C \log N \label{eq54}
	\end{align} 
	where $C > 16D  /c_0$ is a constant.
	Then $\P(\limsup_{N \rightarrow \infty} E_{k,N})=0$. 
\end{lemma}

\begin{proof}
	Since $E_{k,N}$ implies the event that the loss of merging the two segments into one is larger than $f(N)$,
	we obtain from Equality (\ref{eq1}) and the union bound that 
	\begin{align}
		\P(E_{k,N} )
		&\leq 
		\P \biggl\{ \bigcup\limits_{1 \leq n_1 < n_2 < n_3 \leq N_k} 
		\biggl|\sqrt{\frac{n_3-n_2}{n_3-n_1}}\frac{S^{(k)}_{n_1:n_2}}{\sqrt{n_2-n_1}} - \sqrt{\frac{n_2-n_1}{n_3-n_1}}\frac{S^{(k)}_{n_2:n_3}}{\sqrt{n_3-n_2}} \biggr|^2 > f(N) \biggr\}\nonumber \\
		&\leq \sum\limits_{1 \leq n_1 < n_2 < n_3 \leq N_k} 
		\P \biggl\{  
		\biggl|\sqrt{\frac{n_3-n_2}{n_3-n_1}}\frac{S^{(k)}_{n_1:n_2}}{\sqrt{n_2-n_1}} - \sqrt{\frac{n_2-n_1}{n_3-n_1}}\frac{S^{(k)}_{n_2:n_3}}{\sqrt{n_3-n_2}} \biggr|^2 > f(N) \biggr\} \label{eq51}
	\end{align}
	For any tuple $(n_1,n_2,n_3)$, 
	\begin{align}
		\P &\biggl\{  
		\biggl|\sqrt{\frac{n_3-n_2}{n_3-n_1}}\frac{S^{(k)}_{n_1:n_2}}{\sqrt{n_2-n_1}} - \sqrt{\frac{n_2-n_1}{n_3-n_1}}\frac{S^{(k)}_{n_2:n_3}}{\sqrt{n_3-n_2}} \biggr|^2 > f(N) \biggr\} \nonumber \\
		&\leq 
		\P \biggl\{ \bigcup_{d=1}^D \biggl\{  
		\biggl(\sqrt{\frac{n_3-n_2}{n_3-n_1}}\frac{S^{(k)}_{n_1:n_2,d}}{\sqrt{n_2-n_1}} - \sqrt{\frac{n_2-n_1}{n_3-n_1}}\frac{S^{(k)}_{n_2:n_3,d}}{\sqrt{n_3-n_2}} \biggr)^2 > \frac{f(N)}{D} \biggr\} \biggr\} \nonumber \\
		&\leq 
		\sum\limits_{d=1}^D \P \biggl\{  
		\biggl(\sqrt{\frac{n_3-n_2}{n_3-n_1}}\frac{S^{(k)}_{n_1:n_2,d}}{\sqrt{n_2-n_1}} - \sqrt{\frac{n_2-n_1}{n_3-n_1}}\frac{S^{(k)}_{n_2:n_3,d}}{\sqrt{n_3-n_2}} \biggr)^2 > \frac{f(N)}{D} \biggr\}\label{eq50}
	\end{align}
	besides these, from triangular inequality and $n_3-n_2,n_2-n_1<n_3-n_1$, each term in the summation of (\ref{eq50}) is further upper bounded by 
	\begin{align}
		&\P \biggl\{ \bigcup_{\substack{(n^{'},n^{''})=(n_1,n_2)\\\textrm{ or } (n_2,n_3)}} \biggl\{ \biggl| \frac{S^{(k)}_{n^{'}:n^{''},d}}{\sqrt{n^{''}-n^{'}}} \biggl| > \frac{1}{2}\sqrt{\frac{f(N)}{D}}\biggr\} \biggr\}
		\leq \sum\limits_{\substack{(n^{'},n^{''})=(n_1,n_2)\\\textrm{ or } (n_2,n_3)}} \P \biggl\{ \biggl| \frac{S^{(k)}_{n^{'}:n^{''},d}}{n^{''}-n^{'}} \biggl| > \frac{1}{2}\sqrt{\frac{f(N)}{D(n^{''}-n^{'})}}\biggr\} \nonumber \\
		&< 2\exp \biggl\{-c_0 (n^{''}-n^{'}) \frac{1}{4}\frac{f(N)}{D(n^{''}-n^{'})} \biggr\} 
		\leq 2\exp \biggl\{-\frac{c_0 f(N)}{4D }\biggr\} \label{eq52}
	\end{align}
	where the last inequality is due to Assumption (A.4). 
	Combining (\ref{eq50}) and (\ref{eq52}) with (\ref{eq51}), we obtain
	$$
	\P(E_{k,N} ) \leq N_k^3 (2D)\exp \biggl\{-\frac{c_0 f(N)}{4D }\biggr\} 
	\leq 2D N^3 \exp \biggl\{-\frac{c_0 f(N)}{4D }\biggr\}
	\leq 2D N^{-C^{'}}
	$$
	for a constant $C^{'}>1$, where the last inequality follows from (\ref{eq54}).
	Therefore $\sum_{N=1}^{\infty} \P(E_{k,N}) < \infty$ and by Borel-Cantelli lemma $\P(\limsup_{N \rightarrow \infty} E_{k,N})=0$.
\end{proof} 


\begin{remark} 
	Lemma~\ref{lemma:simpleSplit} shows that if there are two neighboring segments that consist of points from the same underlying true segment and if the random variables are sub-Gaussian, then Algorithm~\ref{algo:multiWindow} will almost surely merge them. As a follow up result to Lemma~\ref{lemma:simpleSplit}, Lemma~\ref{lemma:onesideContaminatedSplit} (resp. Lemma~\ref{lemma:twosideContaminatedSplit}) shows that if there are at most $\eta $ points from another true segment from one side (resp. two sides) involved, then Algorithm~\ref{algo:multiWindow} will still merge them almost surely as long as $\eta $ is small in terms of the penalty increment $f(N)$.   	
\end{remark}

\vspace{0.2cm}

\begin{lemma} \label{lemma:onesideContaminatedSplit}
	Suppose that $M_0>0$. For each $k=1,2,\ldots,M_0$, let $E_{k,N}$ denote the event that Algorithm~\ref{algo:oracle} produces two neighboring segments the first of which is a subset of $\{X_n, n \in P_k\}$ and the second of which consists of points from $\{X_n, n \in P_k\}$ and at most $\eta $ points from $\{X_n, n \in P_{k+1}\}$, where $1\leq \eta \leq N_{k+1}$. 
	In other words, 
	\begin{align*}
	E_{k,N} = &\biggl\{ \textrm{There exist integers} \ n_1,n_2,n_3  \textrm{ such that } 0 \leq n_1 < n_2 < N_k, 1\leq n_3 \leq \eta, \textrm{ and }  \\
	&\{X^{(k)}_{n}: n = n_1+1,\ldots,n_2\}, \{X^{(k)}_n: n=n_2+1,\ldots,N_k\} \cup \{X^{(k+1)}_n: n=1,\ldots,n_3\} \\
	& \textrm{ are two detected segments.} \biggr\}
	\end{align*}
	Assume that 
	\begin{align}
		f(N) \geq  \max\{ 16 |\mu_k-\mu_{k+1}|^2 \eta , \ 
		C \log N \}
		\label{eq57}
	\end{align} 
	where $C > 64D  /c_0$ is a constant.
	Then 
	\begin{align} \label{eq74}
		\P\biggl(\limsup_{N\rightarrow \infty} E_{k,N}\biggr)=0.
	\end{align} 
	If we define the event
	\begin{align*}
	\tilde{E}_{k,N} = &\biggl\{ \textrm{There exist integers} \ n_1,n_2,n_3  \textrm{ such that } 1 \leq n_1 < n_2 \leq N_k, 1\leq n_3 \leq \eta, \textrm{ and }  \\
	&\{X^{(k-1)}_{n}: n = N_{k-1}-n_3+1,\ldots,N_{k-1}\}\cup \{X^{(k)}_n: n=1,\ldots,n_1\}, \{X^{(k)}_n: n=n_1+1,\ldots,n_2\} \\
	& \textrm{ are two detected segments.} \biggr\}
	\end{align*}
	where $1\leq \eta \leq N_{k-1}$.
	Assume that $ f(N) \geq  \max\{ 16 |\mu_{k-1}-\mu_{k}|^2 \eta , \ C \log N \} $,
	where $C > 64D  /c_0$ is a constant.
	Then 
	\begin{align} \label{eq75}
		\P\biggl(\limsup_{N\rightarrow \infty} \tilde{E}_{k,N}\biggr)=0.
	\end{align} 
\end{lemma}

\begin{proof}
	We prove (\ref{eq74}). The proof of (\ref{eq75}) is similar.
	Since $E_{k,N}$ implies the event that the loss of merging the two segments into one is larger than $f(N)$
	we obtain from (\ref{eq55}) and the union bound that 
	\begin{align}
		&\P(E_{k,N} )
		\leq 
		\P \biggl\{ \bigcup\limits_{\substack{1 \leq n_1 < n_2 < N_k \\ 1 \leq n_3 \leq \eta}}  
		\biggl|\sqrt{\frac{N_k-n_2+n_3}{N_k-n_1+n_3}}\frac{S^{(k)}_{n_1:n_2}+\mu_k(n_2-n_1)}{\sqrt{n_2-n_1}} - \nonumber \\
		&\qquad \qquad \quad \sqrt{\frac{n_2-n_1}{N_k-n_1+n_3}}\frac{S^{(k)}_{n_2:N_k}+(N_k-n_2)\mu_k + S^{(k+1)}_{0:n_3}+n_3\mu_{k+1} }{\sqrt{N_k-n_2+n_3}} \biggr|^2 > f(N) \biggr\} \nonumber \\
		&\leq 
		\sum\limits_{\substack{1 \leq n_1 < n_2 < N_k \\ 1 \leq n_3 \leq \eta}}  
		\P \biggl\{ \biggl|\sqrt{\frac{N_k-n_2+n_3}{N_k-n_1+n_3}}\frac{S^{(k)}_{n_1:n_2}}{\sqrt{n_2-n_1}} - \sqrt{\frac{n_2-n_1}{N_k-n_1+n_3}}\frac{S^{(k)}_{n_2:N_k} + S^{(k+1)}_{0:n_3} }{\sqrt{N_k-n_2+n_3}} 
		\nonumber \\
		&\qquad \qquad \qquad + const_{n_1,n_2,n_3} \biggr|^2 > f(N) \biggr\} \label{eq56} 
	\end{align}
	where 
	\begin{align*} 
		const_{n_1,n_2,n_3} 
		&= \sqrt{\frac{n_2-n_1}{(N_k-n_1+n_3)(N_k-n_2+n_3)}}n_3(\mu_k-\mu_{k+1}) 
	\end{align*}
	is a constant that depends only on $n_1,n_2,n_3$. 
	Since $n_2-n_1 < N_k-n_1+n_3, n_3 < N_k-n_2+n_3$, we obtain
	\begin{align*}
	| const_{n_1,n_2,n_3}  | 
	< \sqrt{n_3}|\mu_{k+1}-\mu_k|
	\leq \sqrt{\eta}|\mu_{k+1}-\mu_k| = \frac{\sqrt{f(N)}}{4} 
	\end{align*}
	where the last inequality follows from (\ref{eq57}).
	Combining the above result, the inequalities 
	\begin{align*}
		&\sqrt{\frac{N_k-n_2+n_3}{N_k-n_1+n_3}}<1, \quad
	\sqrt{\frac{n_2-n_1}{N_k-n_1+n_3}} \sqrt{\frac{1}{N_k-n_2+n_3}} < \min 
	\biggl\{\frac{1}{\sqrt{N_k-n_2}}, \frac{1}{\sqrt{n_3}} \biggr\}
	\end{align*}
	and Inequality (\ref{eq56}), and using the triangle inequality, we obtain
	\begin{align}
		\P(E_{k,N} )
		&\leq 
		\sum\limits_{\substack{1 \leq n_1 < n_2 < N_k \\ 1 \leq n_3 \leq \eta}} \P \biggl\{ \biggl|\frac{S^{(k)}_{n_1:n_2}}{\sqrt{n_2-n_1}}\biggr|+\biggl|\frac{S^{(k)}_{n_2:N_k}}{\sqrt{N_k-n_2}}\biggr|+\biggl|\frac{S^{(k+1)}_{0:n_3}}{\sqrt{n_3}}\biggr|>\frac{3\sqrt{f(N)}}{4} \biggr\} \label{eq59}
	\end{align}	
	Using the union bound similar to (\ref{eq52}),
	$$
	\P \biggl\{ \biggl|\frac{S^{(k)}_{n_1:n_2}}{\sqrt{n_2-n_1}}\biggr|+\biggl|\frac{S^{(k)}_{n_2:N_k}}{\sqrt{N_k-n_2}}\biggr|+\biggl|\frac{S^{(k+1)}_{0:n_3}}{\sqrt{n_3}}\biggr|>\frac{3\sqrt{f(N)}}{4} \biggr\} 
	$$
	can be upper bounded by
	\begin{align}
	&\sum\limits_{k^{'},n^{'},n^{''}} \P \biggl\{ \biggl| \frac{S^{(k^{'})}_{n^{'}:n^{''}}}{\sqrt{n^{''}-n^{'}}} \biggr| >
		\frac{\sqrt{f(N)}}{4}\biggr\}
		\leq \sum\limits_{k^{'},n^{'},n^{''}}\sum\limits_{d=1}^D \P \biggl\{ \biggl| \frac{S^{(k^{'})}_{n^{'}:n^{''},d}}{n^{''}-n^{'}} \biggr| > \frac{1}{4}\sqrt{\frac{f(N)}{D(n^{''}-n^{'})}} \biggr\} \nonumber \\
	&< 3D \cdot 2 \exp \biggl\{-c_0 (n^{''}-n^{'}) \frac{1}{16}\frac{f(N)}{D(n^{''}-n^{'})} \biggr\}
	\leq 6D\exp \biggl\{- \frac{c_0 f(N)}{16D  } \biggr\} \label{eq58}
	\end{align} 
	where the summation is taken over a tuple $(k^{'},n^{'},n^{''})$ of three possible values: $(k,n_1,n_2)$, $(k,n_2,N_k)$, or $(k+1,0,n_3)$.
	Bringing (\ref{eq58}) into (\ref{eq59}) we obtain
	$$
	\P(E_{k,N} ) \leq N_k^2 \eta (6D)\exp \biggl\{-\frac{c_0 f(N)}{16D }\biggr\} < 6D N^3 \exp \biggl\{-\frac{c_0 f(N)}{16D }\biggr\}
	\leq 6D N^{-C^{'}}
	$$
	for a constant $C^{'}>1$, where the last inequality follows from (\ref{eq57}).
	Therefore $\sum_{N=1}^{\infty} \P(E_{k,N}) < \infty$ and by Borel-Cantelli lemma $\P(\limsup_{N \rightarrow \infty} E_{k,N})=0$.
\end{proof}

\vspace{0.2cm}

\begin{lemma} \label{lemma:twosideContaminatedSplit}
	Suppose that $M_0>1$. For each $k=2,\ldots,M_0$ and $1\leq \eta \leq \min\{N_{k-1},N_{k+1}\}$, define 
	\begin{align*}
	E_{k,N} = &\biggl\{ \textrm{There exist integers} \ n_1,n_2,n_3  \textrm{ such that } 1 \leq n_1 \leq N_k, 1\leq n_2 \leq \eta, 1\leq n_3 \leq \eta, \textrm{ and }  \\
	&\{X^{(k-1)}_{n}: n = N_{k-1}-n_3+1,\ldots,N_{k-1}\}\cup \{X^{(k)}_n: n=1,\ldots,n_1\}, \\
	&\{X^{(k)}_n: n=n_1+1,\ldots,N_k\}\cup \{X^{(k+1)}_n: n=1,\ldots,n_2\} 
	 \textrm{ are two detected segments.} \biggr\}
	\end{align*} 
	Assume that 
	\begin{align}
		f(N) \geq  
		\max\{ 100 |\mu_{k-1}-\mu_{k}|^2 \eta , \ 
		100 |\mu_{k}-\mu_{k+1}|^2 m, \
		C \log N \}
		\label{eq66}
	\end{align} 
	where $C > 100D  /c_0$ is a constant.
	Then $\P(\limsup_{N \rightarrow \infty} E_{k,N})=0$. 
\end{lemma}

\begin{proof}
	The major difference with the proof of Lemma~\ref{lemma:onesideContaminatedSplit} is the treatment of the constant term, which is  
	$$
	const_{n_1,n_2,n_3} = \sqrt{\frac{(n_3+n_1)(n_1'+n_2)}{n_3+N_k+n_2}} \biggl( \frac{n_3 \mu_{k-1}+n_1 \mu_k}{n_3+n_1} - \frac{n_1'\mu_k+n_2\mu_{k+1}}{n_1'+n_2} \biggr)
	$$
	where $n_1'=N_k-n_1$. It can be upper bounded by 
	\begin{align*}
	const_{n_1,n_2,n_3} &= \sqrt{\frac{(n_3+n_1)(n_1'+n_2)}{n_3+N_k+n_2}} \biggl( \frac{n_3}{n_3+n_1}(\mu_{k-1}-\mu_k)+\frac{n_2}{n_1'+n_2}(\mu_{k}-\mu_{k+1}) \biggr) \\
	&\leq \sqrt{\frac{n_3(n_1'+n_2)}{(n_3+n_1)(n_3+N_k+n_2)}} \sqrt{n_3}|\mu_{k-1}-\mu_k|+ \sqrt{\frac{(n_3 +n_1) n_2}{(n_3+N_k+n_2)(n_1'+n_2)}} \sqrt{n_2} |\mu_{k}-\mu_{k+1}| \biggr) \\
	&\leq 2 \sqrt{\eta} \max \{ |\mu_{k-1}-\mu_{k}|, |\mu_{k}-\mu_{k+1}| \} \leq \frac{\sqrt{f(N)}}{5}
	\end{align*}
	The remaining proof is similar to that of Lemma~\ref{lemma:onesideContaminatedSplit}.
\end{proof}

\vspace{0.2cm}

\begin{lemma} \label{lemma:onesideContaminatedSplit_version2}
	Suppose that $M_0>0$. For each $k=1,2,\ldots,M_0$, we define the event 
	\begin{align*}
	E_{k,N} = &\biggl\{ \textrm{There exist integers} \ n_1,n_2,n_3,s  \textrm{ such that } N_k-\eta  \leq n_1 < n_2 < N_k, 1\leq s \leq M_0+1-k, \\
	&\quad   1\leq n_3 \leq N_{k+s},  \textrm{ and } \{X^{(k)}_{n}: n = n_1+1,\ldots,n_2\}, \\
	&\quad \{X^{(k)}_n: n=n_2+1,\ldots,N_k\} \cup \cdots \cup \{X^{(k+s)}_n: n=1,\ldots,n_3\} 
	\textrm{ are two detected segments.} \biggr\}
	\end{align*}
	Assume that 
	\begin{align}
		f(N) \geq  \max\{ (s+3)^2 \bar{\Delta}_{\mu}^2 \eta , \ 
		C \log N \} \label{eq80}
	\end{align} 
	where $C > 4(s+3)^2 D  /c_0$ is a constant.
	Then 
	\begin{align} \label{eq77}
		\P\biggl(\limsup_{N\rightarrow \infty} E_{k,N}\biggr)=0.
	\end{align} 
	If for each $k=2,\ldots,M_0+1$ we define the event
	\begin{align*}
	\tilde{E}_{k,N} = &\biggl\{ \textrm{There exist integers} \ n_1,n_2,n_3  \textrm{ such that } 1 \leq n_1 < n_2 \leq \eta, 1\leq s \leq k-1,   1\leq n_3 \leq N_{k-s}, \textrm{ and }  \\
	&\{X^{(k-s)}_{n}: n = N_{k-s}-n_3+1,\ldots,N_{k-s}\}\cup \cdots \cup \{X^{(k)}_n: n=1,\ldots,n_1\}, \{X^{(k)}_n: n=n_1+1,\ldots,n_2\} \\
	& \textrm{ are two detected segments.} \biggr\}
	\end{align*}
	where $1\leq \eta \leq N_{k}$.
	Assume that 
	\begin{align}
		f(N) \geq  \max\{ (s+3)^2 \bar{\Delta}_{\mu}^2 \eta , \ C \log N \}  \label{eq81}
	\end{align}  
	where $C > 4(s+3)^2 D  /c_0$ is a constant.
	Then 
	\begin{align} \label{eq78}
		\P\biggl(\limsup_{N\rightarrow \infty} \tilde{E}_{k,N} \biggr)=0.
	\end{align} 
\end{lemma}

\begin{proof}
	We prove (\ref{eq77}). The proof of (\ref{eq78}) is similar.
	Similar to Inequality (\ref{eq56}) we obtain 
	\begin{align}
		\P(E_{k,N} )
		&\leq 
		\sum\limits_{\substack{1 \leq n_1 < n_2 < N_k \\ 1 \leq n_3 \leq \eta}}  
		\P \biggl\{ \biggl|\sqrt{\frac{L_{k+s-1}-L_{k-1}-n_2+n_3}{L_{k+s-1}-L_{k-1}-n_1+n_3}}\frac{S^{(k)}_{n_1:n_2}}{\sqrt{n_2-n_1}} \nonumber \\
		&\qquad \qquad  - \sqrt{\frac{n_2-n_1}{L_{k+s-1}-L_{k-1}-n_1+n_3}}\frac{S^{(k)}_{n_2:N_k} + \ldots+ S^{(k+s)}_{0:n_3} }{\sqrt{L_{k+s-1}-L_{k-1}-n_2+n_3}} \nonumber \\
		&\qquad \qquad  + \sqrt{\frac{(L_{k+s-1}-L_{k-1}-n_2+n_3)(n_2-n_1)}{L_{k+s-1}-L_{k-1}-n_1+n_3}}(\mu_{k}-\mu^{*})\biggr|^2 > f(N) \biggr\}  
	\end{align}
	where 
	$$
	\mu^{*} = \frac{(N_k-n_2)\mu_k+\sum\limits_{j=k+1}^{k+s-1}N_{j}\mu_{j}+n_3\mu_{k+s}}{(N_k-n_2)+\sum\limits_{j=k+1}^{k+s-1}N_{j}+n_3}
	$$
	The last term in the above summation can be bounded by
	\begin{align*}
		\biggl| \sqrt{\frac{(L_{k+s-1}-L_{k-1}-n_2+n_3)(n_2-n_1)}{L_{k+s-1}-L_{k-1}-n_1+n_3}}(\mu_{k}-\mu^{*})\biggr|
		&\leq \sqrt{n_2-n_1}|\mu_{k}-\mu^{*}|
		\leq \sqrt{\eta} \bar{\Delta}_{\mu}
		= \frac{\sqrt{f(N)}}{s+3} . 
	\end{align*}
	Following similar proof in Inequalities (\ref{eq59})-(\ref{eq58}), we get 
	$$
	\P(E_{k,N} ) \leq 2(s+3) D N^3 \exp \biggl\{-\frac{c_0 f(N)}{(s+3)^2 D }\biggr\},
	$$
	which implies $\P(\limsup_{N\rightarrow \infty} E_{k,N})=0$
	from Condition (\ref{eq80}) and Borel-Cantelli lemma. 
	Equality (\ref{eq78}) can be similarly proved.
\end{proof}

\begin{remark} 
	Lemma~\ref{lemma:onesideContaminatedSplit_version2} proves that with probability one, for large $N$ there is no detected segment that consists of points from the same true segment while having a small size (compared with the penalty increment $f(N)$).   	
\end{remark}

\vspace{0.2cm}

\begin{lemma} \label{lemma:pureMove}
	Suppose that $M_0>0$. For each $k=1,2,\ldots,M_0$, define the event  
	\begin{align*}
	E_{k,N} = &\biggl\{ \textrm{There exist integers} \ n_1,n_2,n_3  \textrm{ such that } 1 \leq n_1 < n_2 < N_k, 1\leq n_3 \leq N_{k+1}, \textrm{ and }  \\
	&\{X^{(k)}_{n}: n = n_1+1,\ldots,n_2\}, \{X^{(k)}_n: n=n_2+1,\ldots,N_k\} \cup \{X^{(k+1)}_n: n=1,\ldots,n_3\} \\
	& \textrm{ are two detected segments.} \biggr\} 
	\end{align*}
	and the event 
	$A_{k,N} = \bigl\{\min\{N_k-n_2,n_3\} > q_k(N) \bigr\} $
	where
	$ 
		q_k(N) = 250 D  \log N / ( c_0 |\mu_{k}-\mu_{k+1}|^2). 
	$ 
	Then 
	\begin{align}
		\P\biggl\{\limsup_{N\rightarrow \infty} (A_{k,N} \cap E_{k,N}) \biggr\}&=0 \label{eq67}
	\end{align}
	Similarly, if for each $k=2,\ldots,M_0+1$ we define the event
	\begin{align*}
	\tilde{E}_{k,N} = &\biggl\{ \textrm{There exist integers} \ n_1,n_2,n_3  \textrm{ such that } 1 \leq n_1 < n_2 < N_k, 1\leq n_3 \leq N_{k-1}, \textrm{ and }  \\
	&\{X^{(k-1)}_n: n=N_{k-1}-n_3+1,\ldots,N_{k-1}\} \cup \{X^{(k)}_{n}: n = 1,\ldots,n_1\}, \{X^{(k)}_n: n=n_1+1,\ldots,n_2\}  \\
	& \textrm{ are two detected segments.} \biggr\}
	\end{align*} 
	and $\tilde{A}_{k,N} = \bigl\{\min\{n_1,n_3\} > \tilde{q}_k(N) \bigr\}$, $\tilde{q}_k(N) = 250 D  \log N / ( c_0 |\mu_{k-1}-\mu_{k}|^2 )$, then 
	\begin{align}
			\P\biggl\{\limsup_{N\rightarrow \infty} (\tilde{A}_{k,N} \cap \tilde{E}_{k,N})\biggr\}&=0. \label{eq68}
	\end{align}
\end{lemma}

\begin{proof}
	We prove (\ref{eq67}). The proof of (\ref{eq68}) is similar.
	The two detected segments in $E_{k,N}$ contributed to the loss $\mathfrak{L}_1 = Q^{(k)}_{n_1:n_2}+Q^{(k,k+1)}_{n_2:n_3}$. 
	Consider the postulation that the two detected segments are 
	$\{X^{(k)}_{n}: n = n_1+1,\ldots,N_k\}, \{X^{(k+1)}_n: n=1,\ldots,n_3\}$ instead. Correspondingly, its contributed loss is $\mathfrak{L}_2 = Q^{(k)}_{n_1:N_k}+Q^{(k+1)}_{0:n_3}$. 
	Using Equality (\ref{eq_break}) we obtain 
	\begin{align*}
		\mathfrak{L}_1 &= Q^{(k)}_{n_1:n_2}+(Q^{(k)}_{n_2:N_k}+Q^{(k+1)}_{0:n_3}+g^{(k,k+1)}_{n_2,n_3}) \\
		\mathfrak{L}_2 &= (Q^{(k)}_{n_1:n_2}+Q^{(k)}_{n_2:N_k}+g^{(k)}_{n_1,n_2,N_k})+Q^{(k+1)}_{0:n_3}.
	\end{align*}
	where $g^{(k,k+1)}_{n_2,n_3},g^{(k)}_{n_1,n_2,N_k}$ can be expressed as 
	\begin{align}
		g^{(k,k+1)}_{n_2,n_3} &= \biggl|\sqrt{\frac{n_3}{N_k-n_2+n_3}}\frac{S^{(k)}_{n_2:N_k}+(N_k-n_2)\mu_{k}}{\sqrt{N_k-n_2}} - \sqrt{\frac{N_k-n_2}{N_k-n_2+n_3}}\frac{S^{(k+1)}_{0:n_3}+n_3\mu_{k+1}}{\sqrt{n_3}} \biggr| \nonumber \\
		&= \biggl|\sqrt{\frac{n_3}{N_k-n_2+n_3}}\frac{S^{(k)}_{n_2:N_k}}{\sqrt{N_k-n_2}} - \sqrt{\frac{N_k-n_2}{N_k-n_2+n_3}}\frac{S^{(k+1)}_{0:n_3}}{\sqrt{n_3}}  + \sqrt{\frac{n_3(N_k-n_2)}{N_k-n_2+n_3}}(\mu_{k}-\mu_{k+1}) \biggr| \nonumber \\
		g^{(k)}_{n_1,n_2,N_k} &=	 \biggl|\sqrt{\frac{N_k-n_2}{N_k-n_1}}\frac{S^{(k)}_{n_1:n_2}}{\sqrt{n_2-n_1}} - \sqrt{\frac{n_2-n_1}{N_k-n_1}}\frac{S^{(k)}_{n_2:N_k}}{\sqrt{N_k-n_2}} \biggr| \label{eq62}
	\end{align}
	Since event $E_{k,N}$ implies that $\mathfrak{L}_1 \leq \mathfrak{L}_2$, we obtain 
	\begin{align}
		g^{(k,k+1)}_{n_2,n_3} \leq g^{(k)}_{n_2,n_3,N_k}	. \label{eq201}
	\end{align}
	Let $\bar{n}=\min\{N_k-n_2,n_3\}$. From (\ref{eq62}), (\ref{eq201}) and triangle inequality we further obtain   
	\begin{align*}
		\sqrt{\frac{\bar{n}}{2}}|\mu_{k}-\mu_{k+1}| \leq 
		\sqrt{\frac{n_3(N_k-n_2)}{N_k-n_2+n_3}}|\mu_{k}-\mu_{k+1}|
		& \leq \biggl|\frac{S^{(k)}_{n_1:n_2}}{\sqrt{n_2-n_1}} \biggr| + 2\biggl|\frac{S^{(k)}_{n_2:N_k}}{\sqrt{N_k-n_2}} \biggr| + \biggl|\frac{S^{(k+1)}_{0:n_3}}{\sqrt{n_3}} \biggr|
	\end{align*}
Therefore, 
	\begin{align*}
		\P(A_{k,N} \cap E_{k,N})
		&\leq 
		\P \biggl\{ \bigcup\limits_{\substack{1 \leq n_1 < n_2 \leq N_k \\ 1 \leq n_3 \leq N_{k+1}}}
		 \biggl\{ \biggl|\frac{S^{(k)}_{n_1:n_2}}{\sqrt{n_2-n_1}} \biggr| + 2\biggl|\frac{S^{(k)}_{n_2:N_k}}{\sqrt{N_k-n_2}} \biggr| + \biggl|\frac{S^{(k+1)}_{0:n_3}}{\sqrt{n_3}}\biggr|
		 > \sqrt{\frac{\bar{n}}{2}}(\mu_{k}-\mu_{k+1}) , \, \bar{n} > q_k(N) \biggr\} \biggr\}
	\end{align*}
	Using similar techniques as in (\ref{eq52}), we obtain
	\begin{align*}
		\P(A_{k,N} \cap E_{k,N})
		\leq 8D N^3\exp \biggl\{- \frac{c_0 }{16D  } \frac{\bar{n} |\mu_{k}-\mu_{k+1}|^2}{2} \biggr\} 
		< 8D N^{-2}
	\end{align*}
	where the last inequality is due to the definition of $A_{k,N}$.   It is worth mentioning that we will reuse $q_k(N)$ in another lemma, and the constant 250 is  not tight for the purpose of this proof.
	
	Finally, by Borel-Cantelli lemma, we conclude that $\P\{\limsup_{N\rightarrow \infty} (A_{k,N} \cap E_{k,N})\}=0$.

\end{proof}
	
\begin{remark}
	Lemma~\ref{lemma:pureMove} shows that for each $\omega$ from a set of probability one, 
	the event $A_{k,N} \cap E_{k,N}$ will not happen
	for sufficiently large $N$.
	Thus, it shows that each true change point can not be too far away from the detected change point nearest to it. 
	The following Lemma~\ref{lemma:contaminatedMove} is  similar to Lemma~\ref{lemma:pureMove} but in a slightly more complex scenario. 
\end{remark}
	
\vspace{0.2cm}

\begin{lemma} \label{lemma:contaminatedMove}
	Suppose that $M_0>1$ and $\eta $ is an integer that satisfies $1 \leq \eta \leq N_{k-1}$.  For each $k=2,\ldots,M_0$, define the event  
	\begin{align*}
	E_{k,N} = &\biggl\{ \textrm{There exist integers} \ n_1,n_2,n_3  \textrm{ such that } 1 \leq n_1 \leq \eta, 1 \leq n_2 \leq N_k, 1\leq n_3 \leq N_{k+1}, \textrm{ and }  \\
	&\{X^{(k-1)}_{n}: n = N_{k-1}-n_1+1,\ldots,N_{k-1}\} \cup \{X^{(k)}_{n}: n = 1,\ldots,n_2\}, \\
	&\{X^{(k)}_n: n=n_2+1,\ldots,N_k\} \cup \{X^{(k+1)}_n: n=1,\ldots,n_3\} \\
	& \textrm{ are two detected segments.} \biggr\}
	\end{align*}
	and the event 
	\begin{align}
		A_{k,N} = \biggl\{\min \bigl\{N_k-n_2,n_3 \bigr\} \geq \max\biggl\{\frac{4\eta}{(\sqrt{2}-1)^{2}}\frac{|\mu_{k-1}-\mu_{k}|^2}{|\mu_{k}-\mu_{k+1}|^2}, \ 2q_k(N) \biggr\} \biggr\}  \label{eq73},
	\end{align}
	where $q_k(N)$ is the same as was in Lemma~\ref{lemma:pureMove}. Then 
	\begin{align}
		\P\biggl\{\limsup_{N\rightarrow \infty} (A_{k,N} \cap E_{k,N})\biggr\}=0 \label{eq69}
	\end{align} 
	Similarly, if for each $k=2,\ldots,M_0$ we define 
	\begin{align*}
	\tilde{E}_{k,N} = &\biggl\{ \textrm{There exist integers} \ n_1,n_2,n_3  \textrm{ such that } 1 \leq n_1 \leq N_{k-1}, 1 \leq n_2 \leq N_k, 1\leq n_3 \leq \eta, \textrm{ and }  \\
	&\{X^{(k-1)}_{n}: n = N_{k-1}-n_1+1,\ldots,N_{k-1}\} \cup \{X^{(k)}_{n}: n = 1,\ldots,n_2\}, \\
	&\{X^{(k)}_n: n=n_2+1,\ldots,N_k\} \cup \{X^{(k+1)}_n: n=1,\ldots,n_3\} \\
	& \textrm{ are two detected segments.} \biggr\}
	\end{align*}
	and  
	$$\tilde{A}_{k,N} = \biggl\{\min\{n_1,n_2\} \geq \max \biggl\{ \frac{4\eta}{(\sqrt{2}-1)^{2}} \frac{|\mu_{k}-\mu_{k+1}|^{2}}{|\mu_{k-1}-\mu_{k}|^2} , \ 2\tilde{q}_k(N) \biggr\} \biggr\} .$$ 
	where $\tilde{q}_k(N)$ is the same as was in Lemma~\ref{lemma:pureMove}. Then 
	\begin{align}
		\P\biggl\{\limsup_{N\rightarrow \infty} (\tilde{A}_{k,N} \cap \tilde{E}_{k,N}) \biggr\}=0 \label{eq70}
	\end{align} 
\end{lemma}

\begin{proof}
	We prove (\ref{eq69}). The proof of (\ref{eq70}) is similar.
	Similar to the proof of Lemma~\ref{lemma:pureMove}, the event $E_{k,N}$ implies that $\mathfrak{L}_1 \leq \mathfrak{L}_2$ where
	\begin{align*}
		\mathfrak{L}_1 &= Q^{(k-1,k)}_{N_{k-1}-n_1:n_2}+(Q^{(k)}_{n_2:N_k}+Q^{(k+1)}_{0:n_3}+g^{(k,k+1)}_{n_2,n_3}) \\
		\mathfrak{L}_2 &= (Q^{(k-1,k)}_{N_{k-1}-n_1:n_2}+Q^{(k)}_{n_2:N_k}+g^{(k-1,k,k)}_{N_{k-1}-n_1,n_2,N_k})+Q^{(k+1)}_{0:n_3}.
	\end{align*}
	where $g^{(k,k+1)}_{n_2,n_3}$ was given in (\ref{eq62}) and $g^{(k-1,k,k)}_{N_{k-1}-n_1,n_2,N_k}$ can be expressed (similar to (\ref{eq56})) as 
	\begin{align}
		g^{(k-1,k,k)}_{N_{k-1}-n_1,n_2,N_k} &=	 \biggl|\sqrt{\frac{N_k-n_2}{N_k+n_1}}\frac{S^{(k-1)}_{N_{k-1}-n_1:N_{k-1}} +n_1\mu_{k-1}+ S^{(k)}_{0:n_2} + n_2\mu_k}{\sqrt{n_1+n_2}} - \nonumber \\
		&\qquad \qquad \sqrt{\frac{n_1+n_2}{N_k+n_1}}\frac{S^{(k)}_{n_2:N_k}+(N_k-n_2)\mu_k}{\sqrt{N_k-n_2}} \biggr| \nonumber \\
		&\leq 
		 \biggl| \frac{S^{(k-1)}_{N_{k-1}-n_1:N_{k-1}}}{\sqrt{n_1+n_2}} \biggr|
		+ \biggl| \frac{S^{(k)}_{0:n_2}}{\sqrt{n_1+n_2}} \biggr| 
		+ \biggl| \frac{S^{(k)}_{n_2:N_k}}{\sqrt{N_k-n_2}} \biggr| \nonumber \\
		&\quad + \sqrt{\frac{N_k-n_2}{(N_k+n_1)(n_1+n_2)}}n_1 |\mu_{k-1}-\mu_k | \label{eq71}
	\end{align}
	Thus, $\mathfrak{L}_1 \leq \mathfrak{L}_2$ implies that 
	\begin{align}
		g^{(k,k+1)}_{n_2,n_3} \leq g^{(k-1,k,k)}_{N_{k-1}-n_1,n_2,N_k}. \label{eq72}
	\end{align}
	From (\ref{eq62}), (\ref{eq71}), (\ref{eq72}) and triangle inequality we further obtain   
	\begin{align*}
		\sqrt{\frac{n_3(N_k-n_2)}{N_k-n_2+n_3}}|\mu_{k}-\mu_{k+1}|
		& \leq 
		\biggl| \frac{S^{(k-1)}_{N_{k-1}-n_1:N_{k-1}}}{\sqrt{n_1}} \biggr|
		+ \biggl| \frac{S^{(k)}_{0:n_2}}{\sqrt{n_2}} \biggr| 
		+ 2\biggl| \frac{S^{(k)}_{n_2:N_k}}{\sqrt{N_k-n_2}} \biggr|
		+ \biggl| \frac{S^{(k+1)}_{0:n_3}}{\sqrt{n_3}} \biggr|
		\\
		&\quad + \sqrt{\frac{N_k-n_2}{(N_k+n_1)(n_1+n_2)}}n_1 |\mu_{k-1}-\mu_k |
	\end{align*}
Let $\bar{n}=\min\{N_k-n_2,n_3\}$. Since
	\begin{align*}
		&\sqrt{\frac{n_3(N_k-n_2)}{N_k-n_2+n_3}}|\mu_{k}-\mu_{k+1}| \geq \sqrt{\frac{\bar{n}}{2}}|\mu_{k}-\mu_{k+1}|\\
		&\sqrt{\frac{N_k-n_2}{(N_k+n_1)(n_1+n_2)}}n_1 |\mu_{k-1}-\mu_k | \leq \sqrt{n_1}|\mu_{k-1}-\mu_k|
	\end{align*}
and $\sqrt{\bar{n}/2}|\mu_{k}-\mu_{k+1}| - \sqrt{n_1}|\mu_{k-1}-\mu_k| > \sqrt{\bar{n}}|\mu_{k}-\mu_{k+1}|/2 
	\geq  \sqrt{2q_k(N)} 
	|\mu_{k}-\mu_{k+1}|/2$ 
	(from (\ref{eq73})), we obtain 
	\begin{align*}
		&\P(A_{k,N} \cap E_{k,N}) \\
		&\leq 
		\P \biggl\{ \bigcup\limits_{\substack{1 \leq n_1 \leq \eta\\ 1 \leq n_2 \leq N_k \\ 1\leq n_3 \leq N_{k+1} }}
		 \biggl\{ 
		 \biggl| \frac{S^{(k-1)}_{N_{k-1}-n_1:N_{k-1}}}{\sqrt{n_1}} \biggr|
		+ \biggl| \frac{S^{(k)}_{0:n_2}}{\sqrt{n_2}} \biggr| 
		+ 2\biggl| \frac{S^{(k)}_{n_2:N_k}}{\sqrt{N_k-n_2}} \biggr|
		+ \biggl| \frac{S^{(k+1)}_{0:n_3}}{\sqrt{n_3}} \biggr|
		 \geq \frac{\sqrt{2q_k(N) 
		 }}{2}|\mu_{k}-\mu_{k+1}| \biggr\} \biggr\}.
	\end{align*}
	Using similar techniques as in (\ref{eq52}) we obtain
	\begin{align*}
		\P(A_{k,N} \cap E_{k,N})
		\leq 10D N^3\exp \biggl\{- \frac{c_0 q_k(N) |\mu_{k}-\mu_{k+1}|^2}{50D  }\biggr\} 
		= 10 D N^{-2}
	\end{align*}
	Finally, by Borel-Cantelli lemma $\P\{\limsup_{N\rightarrow \infty} (A_{k,N} \cap E_{k,N})\}=0$.
\end{proof}

\vspace{0.5cm}

\textbf{Proof of Theorem~\ref{thm:strongConsistency} (main body):}

	For the case $M_0 = 0$, Lemma~\ref{lemma:simpleSplit} guarantees that there is not overfitting. Next, we prove for the case $M_0>0$. 
	It has been proved in Theorem~\ref{thm:loglog} that there is no underfitting for sufficiently large $N$ almost surely. Note that in its proof, only Assumptions~(A.2)--(A.3) 
	were used. To prove the strong consistency, it remains to prove that there is no overfitting.
	To that end, we define the following sequence of $M_0$ ($M_0>0$) constants $\eta_k(N),k=1,\ldots,M_0$:
	\begin{align}
		\eta_{k}(N) 
		&= \max\biggl\{\frac{4 \eta_{k+1}(N)}{(\sqrt{2}-1)^{2}}\frac{|\mu_{k+1}-\mu_{k+2}|^2}{|\mu_{k}-\mu_{k+1}|^2}, \ 2\tilde{q}_{k+1}(N) \biggr\} , \quad k=1,\ldots,M_0-1 \\
		\eta_{M_0}(N)
		&= \tilde{q}_{M_0+1}(N) 
	\end{align}
	where $\tilde{q}_k(N),k=2,\ldots,M_0+1$ have been defined in Lemmas~\ref{lemma:pureMove}. 
	We prove in three steps sketched below:
	
	Step 1) If Algorithm~\ref{algo:oracle} is applied to $\{ X_{n},n=1,\ldots,N_1+\eta_1 \}$ where $0 \leq \eta_1 \leq \min\{N_2,\eta_1(N)\}$, then almost surely 
		there is no change point detected as $N \rightarrow \infty$. Simply speaking, when the data consists of one true segment and at most $\eta_1(N)$ extra points from another segment at the end, there is no spurious discovery of change points.
		
	Step 2) Suppose that $M_0>1$. If Algorithm~\ref{algo:oracle} is applied to $\{ X_{n}:n=1,\ldots,L_{k}+\eta_k \}$ where $k, \eta_k$ are any integers such that $1 \leq k \leq M_0$ and $0 \leq  \eta_k \leq  \eta_k(N)$, then almost surely there are $k-1$ change point detected, and the largest deviation of each true change point with its nearest detected change point is no larger than $\eta_k(N)$. Simply speaking, when the data consists of $k$ true segments plus at most $\eta_k(N)$ points from the $k+1$th true segment,  the number of true change points $k-1$ is correctly selected.
		
	Step 3) Suppose that $M_0>1$. If Algorithm~\ref{algo:oracle} is applied to $X_{1:N}$, then almost surely there are $M_0$ change point detected. 
	
	Before we prove each step, 
	recall the definitions that $2\tilde{q}_{k} = 500 D  \log N / ( c_0 |\mu_{k-1}-\mu_{k}|^2)$ and $c=4/(\sqrt{2}-1)^{2}$.
	By simple calculation,  we obtain for each $k=1,\ldots,M_0-1$ that
	\begin{align}
		\eta_{k}(N) &= 
		\max \Biggl\{ 
		\bigcup\limits_{\tilde{k}=k,\ldots,M_0-2} 
		\biggl\{ 
		2\tilde{q}_{\tilde{k}+2}(N)
		\prod\limits_{j=k}^{\tilde{k}} \biggl(c\frac{|\mu_{j+1}-\mu_{j+2}|^2}{|\mu_{j}-\mu_{j+1}|^2} \biggr)  
		\biggr\}  \nonumber \\
		&\qquad \quad \cup \bigl\{ 2\tilde{q}_{k+1}(N) \bigr\} \cup 
		\biggl\{ \eta_{M_0}(N)
		\prod\limits_{j=k}^{M_0-1} \biggl(c\frac{|\mu_{j+1}-\mu_{j+2}|^2}{|\mu_{j}-\mu_{j+1}|^2} \biggr)\biggr\}
		\Biggr\} \nonumber \\
		&=\frac{500 D  \log N}{ c_0}
		\max \Biggl\{
		\bigcup\limits_{\tilde{k}=k,\ldots,M_0-2} \biggl\{
		\frac{c^{\tilde{k}-k+1}}{|\mu_{k}-\mu_{k+1}|^2} \biggr\}
		 \cup \biggl\{ \frac{1}{|\mu_{k}-\mu_{k+1}|^2} \biggr\} \cup 
		\biggl\{ \frac{c^{M_0-k}}{2|\mu_{k}-\mu_{k+1}|^2}\biggr\} \Biggr\} \nonumber \\
		&\leq \eta^{*}(N) \label{eq301}
	\end{align}
	where $\eta^{*}(N)$ is defined in Theorem~\ref{thm:strongConsistency}.
	
	\textit{Proof of Step 1):}
	
	If there is at least one change point produced by Algorithm~\ref{algo:oracle}, then its location (in terms of the subscript of $X_n$) belongs to either $\{1,\ldots,N_1\}$ or $\{N_1+1,\ldots,N_1+\eta_1\}$. However, the former case will not happens i.o. due to Lemma~\ref{lemma:onesideContaminatedSplit}, under Condition (\ref{eq57}) (with $E_{k,N}$); and the latter case will not happen i.o. due to Lemma~\ref{lemma:onesideContaminatedSplit_version2} under Condition (\ref{eq81})  (with $\tilde{E}_{k,N}$, $s=1$). We note that Conditions (\ref{eq57}) and (\ref{eq81}) are guaranteed by Inequalities (\ref{condition1}) and~(\ref{eq301}).
	
	\textit{Proof of Step 2):}
	
	Suppose that the last two change points discovered by Algorithm~\ref{algo:oracle} are denoted by $y,z$, i.e. $X_{y+1},\ldots,X_{z}$ and $X_{z+1},\ldots,X_{L_{k}+\eta_k}$ are the last two segments.
		
	The case $k=1$ has been proved in Step 1). Assume that $k>1$ and the statement is true for each $\tilde{k}$ such that $1 \leq \tilde{k} < k$. 
	We prove that the statement holds for $\tilde{k}=k$ as well.
	We consider the three possible events: $z$ belongs to either $\{1,\ldots,L_{k-1}\}$, $\{L_{k-1}+1,\ldots,L_{k}\}$ or $\{L_{k}+1,\ldots,L_{k}+\eta_k(N)\}$, and prove that almost surely $k$ change points are discovered conditioning on each event.
	
	(E1) $z$ belongs to $\{1,\ldots,L_{k-1}\}$. Then, by induction, at most $k-2$ change points are discovered from $\{X_n: n=1,\ldots,z\}$. Thus, there are at most $k-1$ change points in total.
	 
	(E2) $z$ belongs to $\{L_{k-1}+1,\ldots,L_{k}\}$. There are three possible events: (E2.1) $y \leq L_{k-2}$; (E2.2) $L_{k-2}+1 \leq y \leq L_{k-1}$ (E2.3) $L_{k-1}+1 \leq y < z$.
	 
	Given (E2.1), since the  induction guarantees that at most $k-3$ change points are discovered from $\{X_n: n=1,\ldots,y\}$, there are at most $k-1$ change points in total.
	
	Given (E2.2), from Lemma~\ref{lemma:contaminatedMove} (with $\tilde{E}_{k,N}$) and the way $\eta_{k-1}(N)$ was constructed, we obtain $\min\{L_{k-1}-y, \ z-L_{k-1}\} \leq \eta_{k-1}(N)$ for all sufficiently large $N$ almost surely. 
		
	Consider two sub events of (E2.2): (E2.2.1) $1 \leq z-L_{k-1} \leq \eta_{k-1}(N) $, from induction at most $k-2$ change points are discovered from $\{X_n: n=1,\ldots,z\}$, so there are at most $k-1$ change points in total; (E2.2.2) $L_{k-1}-y \leq \eta_{k-1}(N)$, it will not happen i.o. by using Lemma~\ref{lemma:twosideContaminatedSplit} (in which Condition (\ref{eq66}) is guaranteed by (\ref{condition1})). 
	
	For (E2.3), it will not happen i.o. by applying Lemma~\ref{lemma:onesideContaminatedSplit} (with $E_{k,N}$ and Condition (\ref{eq57}) which is guaranteed by (\ref{condition1})).
	
	(E3) $z$ belongs to $\{L_{k}+1,\ldots,L_{k}+\eta_{k}(N)\}$. Four sub events are (E3.1) $y \leq L_{k-2}$, (E3.2) $L_{k-2}+1 \leq y \leq L_{k-1}$, (E3.3) $L_{k-1}+1 \leq y \leq L_{k}$, and (E3.4) $L_{k}+1 \leq y < z$. 
	
	For (E3.1), the induction guarantees that at most $k-2$ change points are discovered from $\{X_n: n=1,\ldots,y\}$, so there are at most $k-1$ change points in total.
	Both the events (E3.2) and (E3.3) will not happen i.o. by  applying Lemma~\ref{lemma:onesideContaminatedSplit_version2} (with $\tilde{E}_{k,N}$, $s=1,2$, and Condition~(\ref{eq81}) which is guaranteed by (\ref{condition1})). 
	By applying Lemma~\ref{lemma:simpleSplit} (with $E_{k+1,N}$), the event (E3.4) will not happen i.o.
	
	\textit{Proof of Step 3):}
	
	Step 3 can be regarded as a special type of step 2 with $k=M_0+1$, and its proof follows from the above proof for events (E1), (E2).
	
	To complete the proof, it remains to prove that the largest deviation of each true change point with its nearest detected change point is less than $\eta^{*}(N)$. 
	This can be proved in similar fashion as above. 

\begin{remark} \label{understand_proof}
	The key part of the proof is Step 2) which builds a induction on $k$, the number of underlying true segments (despite a small amount of extra points). Such induction is achieved through events (E1), (E2.1), (E2.2.1), and (E3.1) at each $k$.	 
	We note that the method differs from the usual mathematical induction in that the number of induction steps is finite, i.e. $k=1,\ldots,M_0$. 	
	Because of that, any (finite) union of events that will not eventually happen will not eventually happen. 
\end{remark}



\ifCLASSOPTIONcaptionsoff
  \newpage
\fi


\bibliographystyle{IEEEtran}
\balance
\bibliography{cp_consistency,Gap,AR_order}
\end{document}